\documentclass[a4paper,12pt]{article}
\usepackage{amsmath}
\usepackage{amssymb}
\usepackage{amsthm}
\usepackage{amsfonts}
\usepackage{mathrsfs}
\usepackage{graphicx}
\usepackage{hyperref}
\usepackage{float}
\newtheorem{lemma}{Lemma}

\newtheorem{preposition}{Preposition}
\newtheorem{assumption}{Assumption}

\numberwithin{equation}{section}

\newcommand{\calE}{\mathsf{E}}

\newtheoremstyle{named}{}{}{\itshape}{}{\bfseries}{.}{.5em}{\thmnote{#3 }#1}
\theoremstyle{named}
\newtheorem*{namedtheorem}{Theorem}

\hypersetup{ pdftitle={Report 2},
pdfauthor={Fiki T. Akbar}, % tuliskan nama
pdfkeywords={Scalar Wave Equation, Non minimal Coupling, Existence}, % tuliskan kata kunci
bookmarksnumbered, pdfstartview={FitH}, urlcolor=blue, }

\begin{document}
\pagestyle{plain}

%%%%%%%%%%%%%%%%%%%%%%%%%%%%%%%%%%%%%%%%%%%%%%%%%%%%%%%%%%%%%%%%%%%%%%%%%%%%%%%%%%%%%%%%%%%%

%% TITLE %%

%%%%%%%%%%%%%%%%%%%%%%%%%%%%%%%%%%%%%%%%%%%%%%%%%%%%%%%%%%%%%%%%%%%%%%%%%%%%%%%%%%%%%%%%%%%%

\title{\LARGE\textbf{General Couplings of Four Dimensional Maxwell-Klein-Gordon System: Global Existence}}

\author{{Mulyanto${^2}$, Fiki Taufik Akbar$^{1,2}$, and Bobby Eka Gunara$^{1,2}$\footnote{Corresponding author} }\\ \\
	$^{1}$\textit{\small Indonesia Center for Theoretical and
		Mathematical Physics (ICTMP)}\\
	$^{2}$\textit{\small Theoretical Physics Laboratory}\\
	\textit{\small Theoretical High Energy Physics and Instrumentation Research Group,}\\
	\textit{\small Faculty of Mathematics and Natural Sciences,}\\
	\textit{\small Institut Teknologi Bandung}\\
	\textit{\small Jl. Ganesha no. 10 Bandung, Indonesia, 40132}\\ \\
	\small email: mulyanto@students.itb.ac.id, ftakbar@fi.itb.ac.id, bobby@fi.itb.ac.id}

\date{}

\maketitle

%%%%%%%%%%%%%%%%%%%%%%%%%%%%%%%%%%%%%%%%%%%%%%%%%%%%%%%%%%%%%%%%%%%%%%%%%%%%%%%%%%%%%%%%%%%%

%% Abstract %%

%%%%%%%%%%%%%%%%%%%%%%%%%%%%%%%%%%%%%%%%%%%%%%%%%%%%%%%%%%%%%%%%%%%%%%%%%%%%%%%%%%%%%%%%%%%%

\begin{abstract}

In this paper, we consider the multi component fields interactions of the complex scalar fields and the electromagnetic fields (Maxwell-Klein-Gordon system) on four dimensional Minkowski spacetime with general gauge couplings and the scalar potential turned on. Moreover, the  complex scalar fields span an internal manifold  assumed to be K\"{a}hler. Then, by taking the  K\"{a}hler potential to be bounded by $U(1)^N$ symmetric  K\"{a}hler potential, the gauge couplings to be bounded functions, and the scalar potential to be the form of either polynomial, sine-Gordon, or Toda potential, we prove the global existence  of the system.

\end{abstract}

%%%%%%%%%%%%%%%%%%%%%%%%%%%%%%%%%%%%%%%%%%%%%%%%%%%%%%%%%%%%%%%%%%%%%%%%%%%%%%%%%%%%%%%%%%%%

% % I. Introduction % %

%%%%%%%%%%%%%%%%%%%%%%%%%%%%%%%%%%%%%%%%%%%%%%%%%%%%%%%%%%%%%%%%%%%%%%%%%%%%%%%%%%%%%%%%%%%%

\section{Introduction}
\label{sec:Introduction}

The Maxwell-Klein-Gordon (MKG) system describes the multi-field interaction between the complex scalar fields $\phi^a$ and the electromagnetic fields $A_{\mu}^\Lambda$ on 4-dimensional Minkowski spacetime with the standard coordinate $x^{\mu} = (t,x^i)$ where $\mu=0,1,2,3$ and $i=1,2,3$, and metric $\eta_{\mu\nu} = \mathrm{diag}(-1,1,1,1)$. Let $\Lambda,\Sigma, \Gamma = 1,2,3, ...,N_V$, denotes the number of gauge fields and the Roman index $a,b, c = 1,2,3, ...,N_C$, describes the number of complex scalar fields. The complex scalar fields span an internal manifold assumed to be Hermitian manifold endowed with metric $g_{a\bar{b}}$.

The Lagrangian of this MKG system with general couplings can be written down as
\begin{equation}\label{langrange}
\mathcal{L}=-\frac{1}{4}{{h}_{\Lambda \Sigma }}\left( \phi ,\overline{\phi } \right)\mathcal{F}_{\mu \nu }^{\Lambda }{{\mathcal{F}}^{\Sigma \left| \mu \nu  \right.}}+\frac{1}{4}{{k}_{\Lambda \Sigma }}\left( \phi ,\overline{\phi } \right)\mathcal{F}_{\mu \nu }^{\Lambda }{{\tilde{\mathcal{F}}}^{\Sigma \left| \mu \nu  \right.}}-{{g}_{a\overline{b}}}\left( \phi ,\overline{\phi } \right){{D}_{\mu }}{{\phi }^{a}}\overline{{{D}^{\mu }}{{\phi }^{b}}}-V\left( \phi ,\bar{\phi } \right) ~ ,
\end{equation}
where $\mathcal{F}_{\mu \nu }^{\Lambda } \equiv {{\partial }_{\mu }}A_{\nu }^{\Lambda }-{{\partial }_{\nu }}A_{\mu }^{\Lambda }$ is the gauge field strength, ${{D}_{\mu }}{{\phi }^{a}} \equiv {{\partial }_{\mu }}{{\phi }^{a}}-i{{q}_{\Gamma }}A_{\mu }^{\Gamma }{{\phi }^{a}}$ is the covariant derivative, ${{q}_{\Gamma}}$ is a constant function, and ${{\tilde{\mathcal{F}}}^{\Sigma \left| \mu \nu  \right.}}$ is the Hodge dual of the $\mathcal{F}_{\mu \nu }^{\Lambda }$. The real functions 
${{h}_{\Lambda \Sigma }}\left( \phi ,\bar{\phi } \right)$ and ${{k}_{\Lambda \Sigma }}\left( \phi ,\bar{\phi } \right)$ denote the gauge couplings, and the real function $V\left( \phi ,\bar{\phi } \right)$ is the scalar potential of MKG system. Since the Lagrangian \eqref{langrange} is invariant under local $U(1)^N$ symmetry where $U(1)^N = \bigotimes_{i=1}^{N}U(1)$, without loss of generality we could take $ h_{\Lambda \Sigma } \left( \phi ,\bar{\phi } \right) = h_{\Lambda \Sigma } (|\phi|^2) $, $k_{\Lambda \Sigma }\left( \phi ,\bar{\phi } \right) = k_{\Lambda \Sigma } (|\phi|^2) $, and  $V\left( \phi ,\bar{\phi } \right) = V (|\phi|^2)$ where $ |\phi|^2 \equiv  \delta_{a\bar{b}} \phi^{a} \bar{\phi}^{\bar{b}}$.  Moreover, to simplify the case, the scalar manifold has to be K\"{a}hler manifold equipped with metric $g_{a\bar{b}} = \partial_{a}\partial_{\bar{b}}K$ where $K\equiv K(\phi,\bar{\phi})$ is the K\"{a}hler potential admitting $U(1)^N$ isometry.

In recent years, the analysis of the global solution of the Maxwell Klein-Gordon system has been considered in detail. These systems obey the local gauge transformation. Then, we have the freedom to choose the gauge condition according to the method and problem considered. The well-known gauge choices are the Coulomb gauge, $\partial^{i} {{A}_{i}}=0$, and the Lorenz gauge, ${{\partial }_{\mu }}{{A}^{\mu }}=0$. Considering the Coulomb gauge condition, Klainerman and Machedon proved that the solution exists globally for the finite initial data \cite{Klainerman}. As for the global well-posedness of the MKG systems in the Lorentz gauge condition, we suggest an interested reader to consult \cite{Selberg}.

Another gauge condition we can take into account is the temporal gauge condition, $A_{0}=0$ as in \cite{Ma, yuan}. In this choice, the local and global existence of the Yang-Mills-Higgs equation is also obtained \cite{Eardley,Moncrief}. In the previous studies, we consider the temporal gauge conditions to prove the local existence of the bosonic part of $N=1$ supersymmetric Yang-Mills-Higgs  with  general couplings and the scalar potential turned on \cite{akbar2}. In particular, we take the K\"ahler potential to be bounded above by the $U(n)$ symmetric K\"ahler potential, the first derivative of the scalar potential to be locally Lipshitz,  and the first derivative of gauge couplings to be at most linear growth functions.

It is of interest to complete the proof of our previous study by showing its global existence. However, we are facing a problem to regularize such as the three vertex terms of the gauge field $A_{\mu}^\Lambda$ when the gauge couplings ${{h}_{\Lambda \Sigma }}\left( \phi ,\bar{\phi } \right)$ and ${{k}_{\Lambda \Sigma }}\left( \phi ,\bar{\phi } \right)$ are no longer constants and moreover, the internal manifold of the $\sigma$-model is not flat. Therefore, to evade such a problem, in this paper we only establish the proof of  the global existence of the MKG system for multi-field interactions with the addition of general gauge coupling in temporal gauge and scalar potential turned on. In particular, we assume that the gauge couplings  ${{h}_{\Lambda \Sigma }}\left( \phi ,\bar{\phi } \right)$ and ${{k}_{\Lambda \Sigma }}\left( \phi ,\bar{\phi } \right)$  are bounded smooth functions, the  K\"{a}hler potential $K$ has to be bounded by $U(1)^N$ symmetric  K\"{a}hler potential which generalizes the case of \cite{akbar2}, and the scalar potential has the form of either polynomial, sine-Gordon, or Toda potential.

We organize the paper as follows. In Section \ref{sec:FieldEquationsOfMotions}, we derive the field equations of motions for both gauge and complex scalar fields. We also discuss some assumptions regarding the gauge coupling functions. In Section \ref{sec:TheInternalManifold}, we provide the discussion of  properties of the scalar internal manifold which is K\"ahler. In Section \ref{sec:Estimates}, we discuss some estimates for the gauge and complex scalar fields which are the important ingredients for proving the global existence. Finally, the final proof and the main theorem are presented in Section \ref{sec:GlobalExistence}. 

%%%%%%%%%%%%%%%%%%%%%%%%%%%%%%%%%%%%%%%%%%%%%%%%%%%%%%%%%%%%%%%%%%%%%%%%%%%%%%%%%%%%%%%%%%%%

% % II. Field Equations of Motions % %

%%%%%%%%%%%%%%%%%%%%%%%%%%%%%%%%%%%%%%%%%%%%%%%%%%%%%%%%%%%%%%%%%%%%%%%%%%%%%%%%%%%%%%%%%%%%

\section{The Field Equations of Motions}
\label{sec:FieldEquationsOfMotions}

Let us first consider the equations of motions of Lagrangian (\ref{langrange})
\begin{equation}\label{eom1}
{{\partial }^{\alpha }}\mathcal{F}_{\alpha \gamma }^{\Sigma }={{h}^{\Lambda \Sigma }}\left\{ i{{q}_{\Lambda }}{{g}_{a\overline{b}}}\left( {{D}_{\gamma }}{{\phi }^{a}}{{\overline{\phi }}^{{\bar{b}}}}-{{\phi }^{a}}\overline{{{D}_{\gamma }}{{\phi }^{b}}} \right)-{{\partial }^{\alpha }}{{h}_{\Lambda \Gamma }}\mathcal{F}_{\alpha \gamma }^{\Gamma }+{{\partial }^{\alpha }}{{k}_{\Lambda \Gamma }}\tilde{\mathcal{F}}_{\alpha \gamma }^{\Gamma } \right\},
\end{equation}
\begin{equation}\label{eom2}
{{D}^{\alpha }}{{D}_{\alpha }}{{\phi }^{{\bar{b}}}}={{g}^{d\bar{b}}}\left( \frac{1}{4}\mathcal{F}_{\alpha \beta }^{\Lambda }{{\partial }_{d}}G_{\Lambda }^{\alpha \beta }-{{\partial }_{d}}{{g}_{a\bar{c}}}{{D}_{\alpha }}{{\phi }^{a}}\overline{{{D}^{\alpha }}{{\phi }^{c}}}-{{\partial }^{\alpha }}{{g}_{d\bar{c}}}\overline{{{D}_{\alpha }}{{\phi }^{c}}}-{{\partial }_{d}}V \right),
\end{equation}
with ${{h}^{\Lambda \Sigma }}$ is the inverse of ${{h}_{\Lambda \Sigma }}$, and
\begin{equation}
G_{\Lambda }^{\alpha \beta }=-{{h}_{\Lambda \Sigma }}{{\mathcal{F}}^{\Sigma \left| \alpha \beta  \right.}}+{{k}_{\Lambda \Sigma }}{{\tilde{\mathcal{F}}}^{\Sigma \left| \alpha \beta  \right.}}.
\end{equation}

The gauge field strength also yields the Bianchi identity
\begin{equation}
{{\partial }_{\gamma }}\mathcal{F}_{\alpha \mu }^{\Sigma }+{{\partial }_{\alpha }}\mathcal{F}_{\mu \gamma }^{\Sigma }+{{\partial }_{\mu }}\mathcal{F}_{\gamma \alpha }^{\Sigma }=0 ~ .
\end{equation}
Taking the covariant divergence of (\ref{eom1}) and use the Bianchi identity of the gauge field strength, we have
\begin{equation}\label{dife1}
\begin{split}
\Box \mathcal{F}_{\mu \gamma }^{\Sigma }&=-i{{q}_{\Lambda }}{{h}^{\Lambda \Sigma }}\left\{ {{\partial }_{\mu }}{{g}_{a\bar{b}}}{{D}_{\gamma }}{{\phi }^{a}}{{{\bar{\phi }}}^{{\bar{b}}}}+{{g}_{a\bar{b}}}{{\partial }_{\mu }}\left( {{D}_{\gamma }}{{\phi }^{a}}{{{\bar{\phi }}}^{{\bar{b}}}} \right) \right\}+{{h}^{\Lambda \Sigma }}\left( {{\partial }^{\alpha }}{{h}_{\Lambda \Gamma }} \right){{\partial }_{\alpha }}\mathcal{F}_{\gamma \mu }^{\Gamma } \\ 
& - {{h}^{\Lambda \Sigma }}\left( {{\partial }^{\alpha }}{{\partial }_{\gamma }}{{h}_{\Lambda \Gamma }} \right)\mathcal{F}_{\alpha \mu }^{\Gamma }+{{h}^{\Lambda \Sigma }}\left( {{\partial }^{\alpha }}{{\partial }_{\gamma }}{{k}_{\Lambda \Gamma }} \right)\widetilde{\mathcal{F}}_{\alpha \mu }^{\Gamma } ~ .  
\end{split}
\end{equation}
In a similar way, we can obtain
\begin{equation}\label{dife2}
\Box{{D}_{\mu }}{{\phi }^{a}}=i{{q}_{\Gamma }}{{\partial }^{\alpha }}\mathcal{F}_{\mu \alpha }^{\Gamma }{{\phi }^{a}}+i{{q}_{\Gamma }}\mathcal{F}_{\mu \alpha }^{\Gamma }{{\partial }^{\alpha }}{{\phi }^{a}}+i{{q}_{\Gamma }}{{\partial }^{\alpha }}\left( A_{\alpha }^{\Gamma }{{\partial }_{\mu }}{{\phi }^{a}} \right)-i{{q}_{\Gamma }}{{\partial }^{\alpha }}\left( A_{\mu }^{\Gamma }{{\partial }_{\alpha }}{{\phi }^{a}} \right)+{{\partial }_{\mu }}{{\partial }^{\alpha }}{{D}_{\alpha }}{{\phi }^{a}} ~ ,
\end{equation}
where $\Box \equiv -\partial_{t}^{2}+\partial_{1}^{2}+\partial_{2}^{2}+\partial_{3}^{2}$ is the d'Alembert operator. 

We translate the coordinate system to the lightcone coordinate system centered at $p$. Then, we can write the equations (\ref{dife1}) and (\ref{dife2}) as the vanishing of a surface integral over the interior of the past light cone $K_p$ from a certain point $p$ to the initial data surface. By using spherical means method, we can write the equations of MKG theory with general coupling in the integral form of its field strength as 
\begin{equation}\label{es1}
\begin{split}
\mathcal{F}_{\mu \gamma }^{\Sigma }&={\mathcal{F}_{\mu \gamma }}^{\Sigma |lin}+\frac{1}{4\pi }\int\limits_{{{K}_{p}}}{rdrd\Omega }\left( -i {{h}^{\Lambda \Sigma }}{{q}_{\Lambda }}\left\{ {{\partial }_{\mu }}{{g}_{a\bar{b}}}{{D}_{\gamma }}{{\phi }^{a}}{{{\bar{\phi }}}^{{\bar{b}}}} + {{g}_{a\bar{b}}}{{\partial }_{\mu }}\left( {{D}_{\gamma }}{{\phi }^{a}}{{{\bar{\phi }}}^{{\bar{b}}}} \right) \right\} \right. \\ 
& {{\left. \left. \text{      } + {{h}^{\Lambda \Sigma }}\left( {{\partial }^{\alpha }}{{h}_{\Lambda \Gamma }} \right){{\partial }_{\alpha }}\mathcal{F}_{\gamma \mu }^{\Gamma } - {{h}^{\Lambda \Sigma }}\left( {{\partial }^{\alpha }}{{\partial }_{\gamma }}{{h}_{\Lambda \Gamma }} \right)\mathcal{F}_{\alpha \mu }^{\Gamma } + {{h}^{\Lambda \Sigma }}\left( {{\partial }^{\alpha }}{{\partial }_{\gamma }}{{k}_{\Lambda \Gamma }} \right)\widetilde{\mathcal{F}}_{\alpha \mu }^{\Gamma } \right) \right|}_{t=-r}} ~ , \\ 
\end{split}
\end{equation}
\begin{equation}\label{es2}
\begin{split}
{{D}_{\mu }}{{\phi }^{a}}&={{D}_{\mu }}{{\phi }^{a}}^{|lin}+ \frac{1}{4\pi }\int\limits_{{{K}_{p}}}{rdrd\Omega } \left( i{{q}_{\Gamma }}{{\partial }^{\alpha }}\mathcal{F}_{\mu \alpha }^{\Gamma }{{\phi }^{a}}+i{{q}_{\Gamma }}\mathcal{F}_{\mu \alpha }^{\Gamma }{{\partial }^{\alpha }}{{\phi }^{a}} \right. \\ 
&{{\left. \left. +i{{q}_{\Gamma }}{{\partial }^{\alpha }}\left( A_{\alpha }^{\Gamma }{{\partial }_{\mu }}{{\phi }^{a}} \right)-i{{q}_{\Gamma }}{{\partial }^{\alpha }}\left( A_{\mu }^{\Gamma }{{\partial }_{\alpha }}{{\phi }^{a}} \right) + {{\partial }_{\mu }}{{\partial }^{\alpha }}{{D}_{\alpha }}{{\phi }^{a}} \right) \right|}_{t=-r}} ~ ,\\ 
\end{split}
\end{equation}
with
\begin{eqnarray}
{\mathcal{F}_{\mu \gamma }}^{\Sigma |lin}=\frac{1}{4\pi }{{\int\limits_{{{S}^{2}}}{d\Omega \left[ {{r}_{0}}\frac{\partial \left\{ \mathcal{F}_{\mu \gamma }^{\Sigma } \right\}}{\partial t}+{{r}_{0}}\frac{\partial \left\{ \mathcal{F}_{\mu \gamma }^{\Sigma } \right\}}{\partial r}+\mathcal{F}_{\mu \gamma }^{\Sigma } \right]}}_{t={{t}_{0}},r={{r}_{0}}}},\\
{{D}_{\mu }}{{\phi }^{a|lin}}=\frac{1}{4\pi }{{\int\limits_{{{S}^{2}}}{d\Omega \left[ {{r}_{0}}\frac{\partial \left\{ {{D}_{\mu }}{{\phi }^{a}} \right\}}{\partial t}+{{r}_{0}}\frac{\partial \left\{ {{D}_{\mu }}{{\phi }^{a}} \right\}}{\partial r}+{{D}_{\mu }}{{\phi }^{a}} \right]}}_{t={{t}_{0}},r={{r}_{0}}}}.
\end{eqnarray}

The energy-momentum tensor of this system is given by
\begin{equation}
{{T}^{\mu \nu }}=\frac{{{h}_{\Lambda \Sigma }}}{2}\left( \mathcal{F}_{\gamma }^{\Lambda |\mu }{{\mathcal{F}}^{\Sigma |\nu \gamma }}+\mathcal{F}_{\gamma }^{\Lambda |\mu }{{{\tilde{\mathcal{F}}}}^{\Sigma |\nu \gamma }} \right)+2{{g}_{a\bar{b}}}{{D}^{\mu }}{{\phi }^{a}}\overline{{{D}^{\nu }}{{\phi }^{b}}}-{{\eta }^{\mu \nu }}{{g}_{a\bar{b}}}{{D}_{\gamma }}{{\phi }^{a}}\overline{{{D}^{\gamma }}{{\phi }^{b}}}-{{\eta }^{\mu \nu }}V ~ .
\end{equation}
It is common to split the field strength $F^{\mu\nu}$ into electric and magnetic component as
\begin{equation}
{{E}^{\Sigma |i}}={{\mathcal{F}}^{\Sigma |0i}},\hspace{1cm}{{H}^{\Sigma |i}}={{\tilde{\mathcal{F}}}^{\Sigma |0i}}=\frac{1}{2}{{\varepsilon }^{ijk}}\mathcal{F}_{jk}^{\Sigma }.
\end{equation}
Thus, we can write the energy function as 
\begin{equation}\label{energy}
{{\mathcal{E}}_{0}}={{\int\limits_{{{B}_{p}}}{{{r}^{2}}drd\Omega \left. \left( \frac{{{h}_{\Lambda \Sigma }}}{2}\left( E_{i}^{\Lambda }{{E}^{\Sigma |i}}+H_{i}^{\Lambda }{{H}^{\Sigma |i}} \right)+{{g}_{a\bar{b}}}{{D}^{0}}{{\phi }^{a}}\overline{{{D}^{0}}{{\phi }^{b}}}+{{g}_{a\bar{b}}}{{D}_{i}}{{\phi }^{a}}\overline{{{D}^{i}}{{\phi }^{b}}}+V \right) \right|}}_{t={{t}_{0}}}},
\end{equation}
where $B_p$ represents a solid sphere in the lightcone coordinate that intersect with $K_p$.

Throughout this paper, we  set the Lagrangian (\ref{langrange}) to be invariant with respect to the local transformation $U(1)^N $ such that we have
\begin{equation}
\begin{split}
{{h}_{\Lambda \Sigma }}\left( U\phi ,\overline{U\phi } \right)& ={{h}_{\Lambda \Sigma }}\left( \phi ,\bar{\phi } \right) ~ , \\ 
{{k}_{\Lambda \Sigma }}\left( U\phi ,\overline{U\phi } \right)&={{k}_{\Lambda \Sigma }}\left( \phi ,\bar{\phi } \right) ~ , \\ 
\text{  }V\left( U\phi ,\overline{U\phi } \right)&=V\left( \phi ,\bar{\phi }\right) ~ . \\ 
\end{split}
\end{equation}
As the result of the gauge ambiguity, we have the freedom to choose the appropriate gauge conditions. In particular, we choose the temporal gauge condition
\begin{equation}\label{temporal}
A_{0}^{\Sigma }\left( x \right)=0 ~ ,
\end{equation}
for all $\Sigma$ which has been shown in \cite{akbar2} that the solutions of \eqref{eom1} and \eqref{eom2} satisfy \eqref{temporal} for all time. We take 
\begin{assumption}
\label{GaugeCouplingAssumptions}
\begin{equation}
\begin{split}
{{h}_{\Lambda \Sigma }}\left( \phi ,\bar{\phi } \right)& ={{h}_{\Lambda \Sigma }}\left( \left| \phi  \right|^2 \right) ~ , \\ 
{{k}_{\Lambda \Sigma }}\left( \phi ,\bar{\phi } \right)&={{k}_{\Lambda \Sigma }}\left( \left| \phi  \right|^2  \right) ~ , 
\end{split}
\end{equation}
with $ \left| \phi  \right|^2 \equiv {{\delta }_{a\bar{b}}}{{\phi }^{a}}{{\bar{\phi} }^{{\bar{b}}}}$, and both ${{h}_{\Lambda \Sigma }}\left( \left| \phi  \right|^2  \right)$ and ${{k}_{\Lambda \Sigma }}\left( \left| \phi  \right|^2 \right)$ are bounded functions for all $\Lambda , \Sigma$.
\end{assumption}

%%%%%%%%%%%%%%%%%%%%%%%%%%%%%%%%%%%%%%%%%%%%%%%%%%%%%%%%%%%%%%%%%%%%%%%%%%%%%%%%%%%%%%%%%%%%

% % III. The Internal Manifold % %

%%%%%%%%%%%%%%%%%%%%%%%%%%%%%%%%%%%%%%%%%%%%%%%%%%%%%%%%%%%%%%%%%%%%%%%%%%%%%%%%%%%%%%%%%%%%
\section{The Internal Manifold}
\label{sec:TheInternalManifold}

 This section is devoted to discuss some properties of the internal scalar manifold which has to be K\"{a}hler  mentioned in section \ref{sec:Introduction}. In particular, we consider the case of the K\"{a}hler potential to be bounded above a $U(1)^N$ symmetric function. Our estimates derived in this section play an important role in  proving the global existence of this system.

First of all, let $K$ be a K\"{a}hler potential satisfying the following condition
\begin{equation}\label{int}
K\left( \phi ,\bar{\phi } \right)\le \Phi \left( \left| \phi  \right| \right) ~ ,
\end{equation}
with
\begin{equation}
\left| \phi  \right|={{\left( {{\delta }_{a\bar{b}}}{{\phi }^{a}}{\bar{\phi }^{{\bar{b}}}} \right)}^{\frac{1}{2}}} ~ .
\end{equation}
Then, we have the following lemma
\begin{lemma}
	\label{scalarManifoldLemma}
	Suppose ${\mathcal{M}}$ is a  K\"{a}hler manifold satisfying \eqref{int}. Let $\Phi$ satisfied the   inequality
	\begin{equation}\label{syarat1}
	\left| \frac{{{Q}'}}{2\left| \phi  \right|} \right|\le \sum\limits_{n=0}^{N}{{{b}_{n}}{{\left| \phi  \right|}^{n}}} ~ ,
	\end{equation} 
	where $b_n > 0$ for all $n$, $Q\left( \left| \phi  \right| \right)=\frac{1}{4\left| \phi  \right|}\left( \Phi ''-\frac{\Phi '}{\left| \phi  \right|} \right)$, and $\Phi '=\frac{\partial \Phi }{\partial \left| \phi  \right|}$. Then, we have 
	\begin{equation}\label{eq:Kahlerineq}
	\begin{split}
	\left| K \right|&\le \sum\limits_{n=0}^{N}{\frac{8{{b}_{n}}}{\left( n+4 \right)\left( n+5 \right)\left( n+6 \right)}}{{\left| \phi  \right|}^{n+6}}+\sum\limits_{n=0}^{N}{\frac{12{{b}_{n}}}{\left( n+2 \right)\left( n+3 \right)\left( n+4 \right)}}{{\left| \phi  \right|}^{n+4}}\\
	&+2{{C}_{1}}{{\left| \phi  \right|}^{3}}+{{C}_{2}}\frac{{{\left| \phi  \right|}^{2}}}{2}+{{C}_3} ~ ,
	\end{split}
	\end{equation}
with $C_i \ge 0$ for all $i=1,2,3$. 
\end{lemma}
\begin{proof}
	If $\tilde{\mathcal{M}}$ is a K\"ahler manifold generated by $\Phi$, then we can write the metric ${{\tilde{g}}_{a\bar{b}}}={{\partial }_{a}}{{\partial }_{{\bar{b}}}}\Phi$ as
	\begin{equation}\label{gab}
	{{\tilde{g}}_{a\bar{b}}}=\frac{\Phi '}{2\left| \phi  \right|}{{\delta }_{a\bar{b}}}+\frac{1}{4\left| \phi  \right|}\left( \Phi ''-\frac{\Phi '}{\left| \phi  \right|} \right){{\delta }_{c\bar{b}}}{{\delta }_{a\bar{c}}}{{\phi }^{c}}{{\phi }^{{\bar{c}}}}
	\end{equation}
	with $\Phi '=\frac{\partial \Phi }{\partial \left| \phi  \right|}$. The inverse and the first derivative of ${{\tilde{g}}_{a\bar{b}}}$ can be written down
	\begin{equation*}
	{{\tilde{g}}^{a\bar{b}}}=\frac{2\left| \phi  \right|}{{{\Phi }'}}{{\delta }^{a\bar{b}}}-\frac{2}{{\Phi }'\left| \phi  \right|}\left( \frac{{\Phi }''-\frac{{{\Phi }'}}{\left| \phi  \right|}}{{\Phi }''+\frac{{{\Phi }'}}{\left| \phi  \right|}} \right){{\phi }^{a}}{{\phi }^{{\bar{b}}}} ~ ,
	\end{equation*}
	\begin{equation*}
	{{\partial }_{c}}{{\tilde{g}}_{a\bar{b}}}=Q\left( {{\delta }_{a\bar{b}}}{{\delta }_{c\bar{d}}}+{{\delta }_{c\bar{b}}}{{\delta }_{a\bar{d}}} \right){{\bar{\phi }}^{{\bar{d}}}}+\frac{Q'}{2\left| \phi  \right|}\left( {{\delta }_{c\bar{f}}}{{\delta }_{e\bar{b}}}{{\delta }_{a\bar{d}}} \right){{\bar{\phi }}^{{\bar{d}}}}{{\phi }^{e}}{{\bar{\phi }}^{{\bar{f}}}} ~ ,
	\end{equation*}
	respectively, where
	\begin{equation*}
	Q\left( \left| \phi  \right| \right)=\frac{1}{4{{\left| \phi  \right|}^{2}}}\left( \Phi ''-\frac{\Phi '}{\left| \phi  \right|} \right) ~ ,
	\end{equation*}
	\begin{equation*}
	\frac{Q'}{2\left| \phi  \right|}=\frac{1}{8{{\left| \phi  \right|}^{3}}}\left( \Phi '''-\frac{3\Phi ''}{\left| \phi  \right|}+\frac{3\Phi '}{{{\left| \phi  \right|}^{2}}} \right) ~ .
	\end{equation*}
	Hence, using the assumption in (\ref{syarat1}) and the integral inequality properties
	\begin{equation*}
	\left| \int{f\left( x \right)dx} \right|\le \int{\left| f\left( x \right) \right|dx} ~ ,
	\end{equation*}
	we obtain
	\begin{equation*}
	\left| Q \right|\le \sum\limits_{n=1}^{N}{\frac{{{b}_{n}}}{n+1}{{\left| \phi  \right|}^{n+1}}}+{{C}_{1}} ~ ,
	\end{equation*}
which implies
	\begin{equation*}
	\begin{split}
	\left| \Phi \right|&\le\sum\limits_{n=0}^{N}{\frac{8{{b}_{n}}}{\left( n+4 \right)\left( n+5 \right)\left( n+6 \right)}}{{\left| \phi  \right|}^{n+6}}+\sum\limits_{n=0}^{N}{\frac{12{{b}_{n}}}{\left( n+2 \right)\left( n+3 \right)\left( n+4 \right)}}{{\left| \phi  \right|}^{n+4}}\\
	&+2{{C}_{1}}{{\left| \phi  \right|}^{3}}+{{C}_{2}}\frac{{{\left| \phi  \right|}^{2}}}{2}+{{C}_3} ~ ,
	\end{split}
	\end{equation*}
with $C_i \ge 0$ for all $i=1,2,3$. Thus, applying  (\ref{int}), we get  \eqref{eq:Kahlerineq}.
\end{proof}
By substituting equation (\ref{gab}) into (\ref{energy}), we can express the energy of the system in the form of
\begin{equation}\label{energi_gab}
\begin{split}
{{\mathcal{E}}_{0}}&=\int\limits_{{{B}_{p}}}{{{r}^{2}}drd\Omega \left( \frac{{{h}_{\Lambda \Sigma }}}{2}\left( E_{i}^{\Lambda }{{E}^{\Sigma |i}}+H_{i}^{\Lambda }{{H}^{\Sigma |i}} \right)+\frac{\Phi '}{2\left| \phi  \right|}{{\delta }_{a\bar{b}}}{{D}_{\mu }}{{\phi }^{a}}\overline{{{D}^{\mu }}{{\phi }^{b}}} \right.} \\ 
& {{\left. +\frac{1}{4\left| \phi  \right|}\left( \Phi ''-\frac{\Phi '}{\left| \phi  \right|} \right){{\delta }_{c\bar{b}}}{{\delta }_{a\bar{c}}}{{\phi }^{c}}{{\phi }^{{\bar{c}}}}{{D}_{\mu }}{{\phi }^{a}}\overline{{{D}^{\mu }}{{\phi }^{b}}}+V \right)}_{t={{t}_{0}}}} ~ .  
\end{split}
\end{equation}
Since the energy of the system must be positive, then we should take
\begin{assumption}
\label{boundpotAssumptions}
The function $\Phi$ has to be bounded below by
\begin{equation}\label{eq:energi_gabcon}
\left| \Phi  \right|\ge \frac{{{c}_{1}}}{2}{{\left| \phi  \right|}^{2}}+{{c}_{2}} ~ ,
\end{equation}
where $c_1 \ge 0$ and $c_2 \ge 0$.
\end{assumption}

\section{Estimates}
\label{sec:Estimates}

In this section, we derive some estimates for the complex scalar fields and the gauge fields which are the significant part of the global existence proof.

\subsection{The Flat Energy Estimate}
We define a flat energy functional as
\begin{equation}\label{modE}
{\mathcal{J}}\left( t \right) \equiv {{\left\| E \right\|}_{{{L}^{2}}}}+{{\left\| H \right\|}_{{{L}^{2}}}}+\frac{{{c}_{1}}}{2}{{\left\| D\phi  \right\|}_{{{L}^{2}}}}+{{\left\| \phi  \right\|}_{{{L}^{2}}}}+{{\left\| V \right\|}_{{{L}^{2}}}} ~ ,
\end{equation}
where $c_1 \ge 0$ and
\begin{equation}
{{\left\| f \right\|}_{{{L}^{p}}}}\equiv {{\left( \int\limits_{S}{{{\left| f \right|}^{p}}d\mu } \right)}^{{}^{1}/{}_{p}}} ~ ,
\end{equation}
is a standard $L^p$ norm. The flat energy functional  (\ref{modE}) plays an important role for bounding some estimates in the MKG system. 
\begin{preposition}
	Let ${{\phi }^{a}},E_{i}^{\Lambda },H_{i}^{\Lambda }$ be  solutions of MKG system as in (\ref{eom1}) and (\ref{eom2}) in temporal gauge (\ref{temporal}) with
	\begin{equation}
	{\mathcal{J}_{0}} \equiv {\mathcal{J}}\left( 0 \right)<\infty ~ .	
	\end{equation} 
	Then, for all $t\ge 0$ there exists a positive constant ${{C}_{N}} > 1$ such that
	\begin{equation}\label{preposisi}
	{\mathcal{J}}\left( t \right)\le {{C}_{N}}{\mathcal{J}_{0}}\left( 1+t \right) ~ .
	\end{equation}
	\begin{proof}
		Clearly that ${{\mathcal{E}}_{0}}\le \mathcal{J}_{0}^{2}$. So, the first, second, and last terms of (\ref{modE}) are bounded by $\mathcal{E}_{0}^{1/2}$, that is,
		\begin{equation}
		{{\left\| E \right\|}_{{{L}^{2}}}}+{{\left\| H \right\|}_{{{L}^{2}}}}+{{\left\| V \right\|}_{{{L}^{2}}}}\le \mathcal{E}_{0}^{1/2}\le C{\mathcal{J}_{0}} ~ ,
		\end{equation}
where $C \ge 1$, and we have defined $\left| E \right| \equiv {{\left( E_{i}^{\Lambda }{{E}^{\Sigma |i}} \right)}^{1/2}}$ and $\left| H \right| \equiv {{\left( H_{i}^{\Lambda }{{H}^{\Sigma |i}} \right)}^{1/2}}$.
		Next, from (\ref{eq:energi_gabcon}), we obtain
		\begin{equation}
		\left| \frac{\Phi '}{2\left| \phi  \right|} \right|\ge \frac{{{c}_{1}}}{2} ~ ,
		\end{equation}
		such that the estimate of the third term in (\ref{modE}) has the form
		\begin{equation}
		\frac{{{c}_{1}}}{2}{{\left\| D\phi  \right\|}_{{{L}^{2}}}}\le {\mathcal{J}_{0}} ~ ,
		\end{equation}
		with $\left| D\phi  \right| \equiv {{\left( {{\delta }_{a\bar{b}}}{{D}_{\mu }}{{\phi }^{a}}\overline{{{D}^{\mu }}{{\phi }^{b}}} \right)}^{1/2}}$.
		To get the estimate of the fourth term in (\ref{modE}), let us consider
		\begin{equation}
		\begin{split}
		\frac{\partial }{\partial t}\int\limits_{{{K}_{p}}}{{{r}^{2}}drd\Omega }{{\left| \phi  \right|}^{2}}& \text{ = 2 Re}\int\limits_{{{K}_{p}}}{{{r}^{2}}drd\Omega }\left| \phi {{D}_{0}}\phi  \right| \\ 
		& \le \text{2}\left\| \phi  \right\|_{{{L}^{2}}}^{{}}\left\| {{D}_{0}}\phi  \right\|_{{{L}^{2}}}^{{}}\text{ }\le \text{2}{\mathcal{J}_{0}}\left\| \phi  \right\|_{{{L}^{2}}} ~ .  \\ 
		\end{split}
		\end{equation}
		Integrating the inequality, we obtain
		\begin{equation}
		{{\left\| \phi  \right\|}_{{{L}^{2}}}}\le \tilde{C} {\mathcal{J}_{0}}\left( 1+t \right) ~ ,
		\end{equation}
	with $ \tilde{C} > 0$	showing that (\ref{preposisi}) is fulfilled. 
	\end{proof} 
\end{preposition}

\subsection{Estimate for the gauge fields}
Let us rewrite the integral equation (\ref{es1}) as
\begin{equation}\label{eq:es1lg}
\mathcal{F}_{\mu \gamma }^{\Sigma } ={\mathcal{F}_{\mu \gamma }}^{\Sigma |lin} + I_{\mu \gamma }^{\Sigma } + J_{\mu \gamma }^{\Sigma } + K_{\mu \gamma }^{\Sigma } + L_{\mu \gamma }^{\Sigma } + M_{\mu \gamma }^{\Sigma } ~ ,
\end{equation}
where
\begin{equation}\label{eq:nonlinI}
 I_{\mu \gamma }^{\Sigma }  \equiv  - \frac{1}{4\pi} \int\limits_{{{K}_{p}}}{rdrd\Omega } {{\left. h^{\Lambda \Sigma } {{\partial }^{\alpha }}{{\partial }_{\gamma }}{{h}_{\Lambda \Gamma }}\mathcal{F}_{\alpha \mu }^{\Gamma } \right|}_{t=-r}} ~ ,
\end{equation}
\begin{equation}\label{eq:nonlinJ}
 J_{\mu \gamma }^{\Sigma }  \equiv   \frac{1}{4\pi} \int\limits_{{{K}_{p}}}{rdrd\Omega } {{\left. {{h}^{\Lambda \Sigma }}\left( {{\partial }^{\alpha }}{{h}_{\Lambda \Gamma }} \right){{\partial }_{\alpha }}\mathcal{F}_{\gamma \mu }^{\Gamma } \right|}_{t=-r}} ~ ,
\end{equation}
\begin{equation}\label{eq:nonlinK}
 K_{\mu \gamma }^{\Sigma }  \equiv   - \frac{i}{4\pi} \int\limits_{{{K}_{p}}}{rdrd\Omega } {{\left.{{h}^{\Lambda \Sigma }}{{q}_{\Lambda }} {{\partial }_{\mu }}{{g}_{a\bar{b}}}{{D}_{\gamma }}{{\phi }^{a}}{{{\bar{\phi }}}^{{\bar{b}}}} \right|}_{t=-r}} ~ ,
\end{equation}
\begin{equation}\label{eq:nonlinL}
 L_{\mu \gamma }^{\Sigma }  \equiv   - \frac{i}{4\pi}  \int\limits_{{{K}_{p}}}{rdrd\Omega } {{\left. {{h}^{\Lambda \Sigma }}{{q}_{\Lambda }} {{g}_{a\bar{b}}}{{\partial }_{\mu }}\left( {{D}_{\gamma }}{{\phi }^{a}}{{{\bar{\phi }}}^{{\bar{b}}}} \right)  \right|}_{t=-r}} ~ ,
\end{equation}
\begin{equation}\label{eq:nonlinM}
 M_{\mu \gamma }^{\Sigma }  \equiv   \frac{1}{4\pi} \int\limits_{{{K}_{p}}}{rdrd\Omega } {{\left. {{h}^{\Lambda \Sigma }}\left( {{\partial }^{\alpha }}{{\partial }_{\gamma }}{{k}_{\Lambda \Gamma }} \right)\widetilde{\mathcal{F}}_{\alpha \mu }^{\Gamma }  \right|}_{t=-r}} ~ .
\end{equation}

First, we consider the first term of the nonlinear part of the equation in (\ref{eq:es1lg}) which satisfies
\begin{equation}
\left| I_{\mu \gamma }^{\Sigma } \right| \le  \int\limits_{{{K}_{p}}}{rdrd\Omega }{{\left| {{h}^{\Lambda \Sigma }}{{\partial }^{\alpha }}{{\partial }_{\gamma }}{{h}_{\Lambda \Gamma }}\mathcal{F}_{\alpha \mu }^{\Gamma } \right|}_{t=-r}} ~ .
\end{equation}
 Using the Holder inequality, we have the following estimate
\begin{equation*}
\begin{split}
\left| I_{\mu\gamma}^{\Sigma } \right| & \le  {{\left( \int\limits_{{{K}_{p}}}{{{r}^{2}}drd\Omega \frac{{{\left| \mathcal{F}_{\alpha\mu }^{\Gamma } \right|}^{2}}}{{{r}^{2}}}} \right)}^{1/2}}{{\left( \int\limits_{{{K}_{p}}}{{{r}^{2}}drd\Omega {{\left| {{h}^{\Lambda \Sigma }}{{\partial }^{\alpha}}{{\partial }_{\gamma}}{{h}_{\Lambda \Gamma }} \right|}^{2}}} \right)}^{1/2}} \\ 
& \le  {{\left( \int\limits_{0}^{{{r}_{0}}}{dr}\left\|\mathcal{F}(-r) \right\|_{{{L}^{\infty }}}^{2} \right)}^{1/2}}{{\left\| {{h}^{\Lambda \Sigma }}{{\partial }^{\alpha}}{{\partial }_{\gamma }}{{h}_{\Lambda \Gamma }} \right\|}_{{{L}^{2}}}} ~ .\\ 
\end{split}
\end{equation*}
We have then
\begin{equation}\label{eq1}
{{\left\| {{h}^{\Lambda \Sigma }}{{\partial }^{\alpha }}{{\partial }_{\gamma}}{{h}_{\Lambda\Gamma }} \right\|}_{{{L}^{2}}}}\le c\left( {{\left\| {{\partial }^{\alpha}}\Psi \text{ }{{\partial }_{\gamma}}\Psi  \right\|}_{{{L}^{2}}}}+{{\left\| {{\partial }^{\alpha }}{{\partial }_{\gamma}}\Psi  \right\|}_{{{L}^{2}}}} \right) ~ ,
\end{equation}
where $c \ge 0$ and we have defined $\Psi  \equiv |\phi|^2 = {{\delta }_{a\bar{b}}}{{\phi }^{a}}{\bar{\phi }^{{\bar{b}}}}$.

Consider the following estimate
\begin{equation*}
\begin{split}
{{\left\| \phi  \right\|}_{{{L}^{3}}}}&\le \left\| \phi  \right\|_{{{L}^{2}}}^{1/2}\left\| \nabla \phi  \right\|_{{{L}^{2}}}^{1/2} \\ 
& \le {{c}_{N}}\mathcal{J}_{0}^{1/2}{{\left( 1+t \right)}^{1/2}}\left\| \nabla \phi  \right\|_{{{L}^{2}}}^{1/2} \\ 
& \le {{c}_{N}}\mathcal{J}_{0}^{1/2}{{\left( 1+t \right)}^{1/2}}{{\left( {{\left\| {{D}_{i}}\phi  \right\|}_{{{L}^{2}}}}+{{\left\| \nabla A_{i}^{{}} \right\|}_{{{L}^{2}}}}{{\left\| \phi  \right\|}_{{{L}^{3}}}} \right)}^{1/2}}\text{ } \\ 
&\le {{c}_{N}}\mathcal{J}_{0}^{{}}\left( 1+t \right){{\left( 1+{{\left\| \phi  \right\|}_{{{L}^{3}}}} \right)}^{1/2}} \\ 
&\le {{c}_{N}}\mathcal{J}_{0}^{{}}\left( 1+t \right)  ~ ,
\end{split}	
\end{equation*}
with $c_N > 1$. Using the results above, we can get the following estimate
\begin{equation*}
\begin{split}
{{\left\| \nabla \phi  \right\|}_{{{L}^{2}}}}&\le {{\left\| {{D}_{i}}\phi  \right\|}_{{{L}^{2}}}}+{{\left\| \nabla {{A}_{i}} \right\|}_{{{L}^{2}}}}{{\left\| \phi  \right\|}_{{{L}^{3}}}} \\ 
&\le {{\mathcal{J}}_{0}}+{{\mathcal{J}}_{0}}{{\left\| \phi  \right\|}_{{{L}^{3}}}} \\ 
&\le {{C}_{N}}\mathcal{J}_{0}^{2}\left( 1+t \right) ~ . \\ 
\end{split}
\end{equation*}
Thus, the first term of (\ref{eq1}) can be bound to the energy and its $L^\infty$ norm
\begin{equation}
\begin{split}
{{\left\| {{\partial }^{\alpha}}\Psi \text{ }{{\partial }_{\gamma }}\Psi  \right\|}_{{{L}^{2}}}}&\le {{\left\| {{\partial }^{\alpha}}\Psi \text{ } \right\|}_{{{L}^{\infty }}}}{{\left\| {{\partial }_{\gamma}}\Psi  \right\|}_{{{L}^{2}}}} \\ 
& \le \tilde{C}_N {{\left\| {{\partial }^{\alpha}}\Psi \text{ } \right\|}_{{{L}^{\infty }}}}{{\left\| \phi  \right\|}_{{{L}^{\infty }}}}{{\left\| {{\partial }_{\gamma}}\phi  \right\|}_{{{L}^{2}}}} \\
& \le \tilde{C}_N \mathcal{J}_0^2(1+t)\left\|\partial^\alpha \Psi\right\|_{L^\infty}\|\phi\|_{L^\infty} ~ ,
\end{split}
\end{equation}
with $\tilde{C}_N > 1$.

As for the second term ${{\left\| {{\partial }^{\alpha}}{{\partial }_{\gamma}}\Psi  \right\|}_{{{L}^{2}}}}$, we have 
\begin{equation}
{{\left\| {{\partial }^{\alpha}}{{\partial }_{\gamma}}\Psi  \right\|}_{{{L}^{2}}}} \le  {{C}_{N}}\mathcal{J}_{0}^{2}\left( 1+t \right){{\left\| \partial \phi  \right\|}_{{{L}^{\infty }}}}+{{\left\| \phi  \right\|}_{{{L}^{\infty }}}}{{\left\| {{\partial }^{\alpha}}{{\partial }_{\gamma}}\phi  \right\|}_{{{L}^{2}}}}   ~ ,
\end{equation}
where
\begin{equation}
{{\left\| {{\partial }^{\alpha}}{{\partial }_{\gamma}}{{\phi }} \right\|}_{{{L}^{2}}}}\le {{\left\| {{D}^{\alpha }}{{D}_{\alpha }}  \phi \right\|}_{{{L}^{2}}}} ~ .
\end{equation}
We can use the equation of motion in (\ref{eom2}) so that
\begin{equation}
\begin{split}\label{eq:eomineq}
{{\left\| {{D}^{\alpha }}{{D}_{\alpha }}\overline{{{\phi }^{b}}} \right\|}_{{{L}^{2}}}}&\le {{\left\| {{g}^{c\bar{b}}}\mathcal{F}_{\alpha \beta }^{\Lambda }{{\partial }_{c}}G_{\Lambda }^{\alpha \beta } \right\|}_{{{L}^{2}}}}+{{\left\| {{g}^{d\bar{b}}}{{\partial }_{d}}{{g}_{a\bar{c}}}{{D}_{\alpha }}{{\phi }^{a}}\overline{{{D}^{\alpha }}{{\phi }^{c}}} \right\|}_{{{L}^{2}}}}\\
&+{{\left\| {{g}^{d\bar{b}}}{{\partial }^{\alpha }}{{g}_{d\bar{c}}}\overline{{{D}_{\alpha }}{{\phi }^{c}}} \right\|}_{{{L}^{2}}}}+{{\left\| {{g}^{d\bar{b}}}{{\partial }_{d}}V \right\|}_{{{L}^{2}}}}  ~ .
\end{split}
\end{equation}
In order to have an estimate of \eqref{eq:eomineq}, we have to specify the form of  the scalar potential $V\left(\phi,\bar{\phi}\right)$: 
\begin{assumption}
\label{ScalarpotentialAssumptions}
 The scalar potential  $V\left(\phi,\bar{\phi}\right)$ has to be either of the following form
\begin{equation}\label{potential1}
V\left( \Psi  \right)=\sum\limits_{n=0}^{\tilde{N}}{{{a}_{n}}{{\Psi }^{n}}} ~ ,
\end{equation}
\begin{equation}\label{potential2}
V\left( \Psi \right)= V_0 \left( 1-\cos \lambda \Psi  \right) ~ ,
\end{equation}
\begin{equation}\label{potential3}
 V\left( \Psi  \right)=\sum\limits_{n=0}^{{\tilde{N}}}{{{{\tilde{a}}}_{n}}}{{e}^{-{{{\tilde{\lambda }}}_{n}}\Psi }} ~ ,
\end{equation}
where $a_n, \tilde{a}_n, V_0, \lambda$ are real constants, while $\tilde{\lambda}_n > 0$ for every $n$.
\end{assumption}

It is worth mentioning that there three known examples in the case of \eqref{potential1}, namely, for $\tilde{N} =1$ it corresponds to the mass term in the Klein-Gordon equation, while   for $\tilde{N} =2$ it describes  the ${{\phi }^{4}}$-theory.  Equations \eqref{potential2} and \eqref{potential3} correspond to  the sine-Gordon and the  Toda field theories, respectively.
%%%

Making use of the Holder inequality, equation (\ref{gab}), and its derivative, we have then
\begin{equation}
\begin{split}
{{\left\| {{D}^{\alpha }}{{D}_{\alpha }}\overline{{{\phi }^{b}}} \right\|}_{{{L}^{2}}}} &\le \tilde{c}_1 ~ {{\mathcal{J}}_{0}}\left( 1+{{\left\| \partial \phi  \right\|}_{{{L}^{\infty }}}} \right)\left( {{\left\| \phi  \right\|}_{{{L}^{\infty }}}}+\left\| \phi  \right\|_{{{L}^{\infty }}}^{3}+\sum\limits_{n=1}^{N}{\left\| \phi  \right\|_{{{L}^{\infty }}}^{n+2}}+\sum\limits_{n=1}^{N}{\left\| \phi  \right\|_{{{L}^{\infty }}}^{n+4}} \right) \\ 
& + \tilde{c}_2 ~ {{\mathcal{J}}_{0}}{{\left\| \phi  \right\|}_{{{L}^{\infty }}}}{{\left\| \partial \phi  \right\|}_{{{L}^{\infty }}}}+  \tilde{c}_3 ~ {{\mathcal{J}}_{0}} ~ {\mathcal{I}}\left( \left\| \phi  \right\| \right)  ~ ,
\end{split}
\end{equation}
with $\tilde{c}_i   > 1$ for all $i=1,2,3$, and
\begin{equation}
{\mathcal{I}} \left( \left\| \phi  \right\| \right) \le \begin{cases}
{\mathcal{O}} \left( \left\| \phi  \right\| \right) ~ , &\text{if $V$ is of the form \eqref{potential1}} ~ ,\\
{{\left\| \phi  \right\|}_{{{L}^{\infty }}}}{{\left\| \partial \phi  \right\|}_{{{L}^{\infty }}}}{{\mathcal{J}}_{0}} ~ , &\text{if $V$ is either of the form \eqref{potential2} or \eqref{potential3}} ~ ,
\end{cases}
\end{equation}
where
\begin{equation}
{\mathcal{O}} \left( \left\| \phi  \right\| \right)\le {{\left\| \phi  \right\|}_{{{L}^{\infty }}}}{{\left\| \partial \phi  \right\|}_{{{L}^{\infty }}}}\left( 1+{{\mathcal{J}}_{0}}\left( 1+t \right)\sum\limits_{n=1}^{N-2}{\left\| \phi  \right\|_{{{L}^{\infty }}}^{2n+1}} \right) ~ .
\end{equation}
Thus, we obtain the estimate for the first term of   (\ref{eq:es1lg})
\begin{equation}
\begin{split}
\left| I_{\mu \gamma }^{\Sigma } \right| &\le {{\mathcal{C}}_1 } \mathcal{J}_{0}^{2}\left( 1+t \right){{\left( \int\limits_{0}^{{{r}_{0}}}{dr}\left\| \mathcal{F}(-r) \right\|_{{{L}^{\infty }}}^{2} \right)}^{1/2}}\left\{ {{\left\| \partial \Psi  \right\|}_{{{L}^{\infty }}}}{{\left\| \phi  \right\|}_{{{L}^{\infty }}}} \right.+{{\left\| \partial \phi  \right\|}_{{{L}^{\infty }}}} \\ 
& +\left( {{\left\| \phi  \right\|}_{{{L}^{\infty }}}}+\left\| \phi  \right\|_{{{L}^{\infty }}}^{3}+\sum\limits_{n=1}^{N}{\left\| \phi  \right\|_{{{L}^{\infty }}}^{n+2}}+\sum\limits_{n=1}^{N}{\left\| \phi  \right\|_{{{L}^{\infty }}}^{n+4}} \right)\left( {{\left\| \partial \phi  \right\|}_{{{L}^{\infty }}}}+1 \right) \\ 
& \left. +\left\| \phi  \right\|_{{{L}^{\infty }}}^{2}{{\left\| \partial \phi  \right\|}_{{{L}^{\infty }}}}+{{\left\| \phi  \right\|}_{{{L}^{\infty }}}}\mathcal{I}\left( \left\| \phi  \right\| \right) \right\}  ~ ,
\end{split}
\end{equation}
where $ {{\mathcal{C}}_1} > 1$ and
\begin{equation}
{{\left\| ^{(4)}\mathcal{F}(t) \right\|}_{{{L}^{\infty }}}} \equiv \left| \mathcal{F}_{\alpha \beta }^{\Sigma }\mathcal{F}_{{}}^{\Sigma |\alpha \beta }\left( t \right) \right|_{L^\infty }^{1/2} ~ .
\end{equation}

For the second term of nonlinear part in  (\ref{eq:es1lg}), we have
\begin{equation}
\left| J_{\mu \gamma }^{\Sigma } \right| \le \int\limits_{{{K}_{p}}}{rdrd\Omega }{{\left| {{h}^{\Lambda \Sigma }}\left( {{\partial }^{\alpha }}{{h}_{\Lambda \Gamma }} \right){{\partial }_{\alpha }}\mathcal{F}_{\gamma \mu }^{\Gamma } \right|}_{t=-r}} \:. \\ 
\end{equation}
Making use of the Holder inequality again, we get
\begin{equation}
\left| J_{\mu \gamma }^{\Sigma } \right| \le  {{\left( \int\limits_{0}^{{{r}_{0}}}{dr}\left\| D\phi  \right\|_{{{L}^{\infty }}}^{2} \right)}^{1/2}}{{\left\| \phi  \right\|}_{{{L}^{\infty }}}}{{\left\| {{\partial }_{\alpha }}\mathcal{F}_{\gamma \mu }^{\Gamma } \right\|}_{{{L}^{2}}}}  \;,
\end{equation}
where 
\begin{equation} 
{{\left\| D\phi \right\|}_{{{L}^{\infty }}}} \equiv \left| \delta_{a\bar{b}}{{D}_{\alpha }}{{\phi }^{a}} \overline{{D}^{\alpha }{\phi }^{b}} \right|_{L^\infty }^{1/2}.
\end{equation}
We can use the equation of motion in (\ref{eom1}) to get estimate for ${{\left\| {{\partial }_{\alpha }}\mathcal{F}_{\gamma \mu }^{\Gamma } \right\|}_{{{L}^{2}}}}$, namely,
\begin{equation}\label{es eom 1}
{{\left\| {{\partial }_{\alpha }}\mathcal{F}_{\gamma \mu }^{\Gamma } \right\|}_{{{L}^{2}}}}\le {\mathcal{C}} {{\mathcal{J}}_{0}}\left( {{\left\| \partial \Psi  \right\|}_{{{L}^{\infty }}}}+{{\left\| \phi  \right\|}_{{{L}^{\infty }}}} \right) ~ ,
\end{equation}
with ${\mathcal{C}} > 1$. Thus, the estimate for the second term is given by
\begin{equation}
\left| J_{\mu \gamma }^{\Sigma } \right| \le {\mathcal{C}_{2}} {{\mathcal{J}}_{0}}{{\left( \int\limits_{0}^{{{r}_{0}}}{dr}\left\| D\phi  \right\|_{{{L}^{\infty }}}^{2} \right)}^{1/2}}\left( {{\left\| \phi  \right\|}_{{{L}^{\infty }}}}{{\left\| \partial \Psi  \right\|}_{{{L}^{\infty }}}}+\left\| \phi  \right\|_{{{L}^{\infty }}}^{2} \right) ~ ,
\end{equation}
with ${\mathcal{C}}_2 > 1$. 

Similarly, by employing some computation, we obtain the estimates for the third and fourth term of (\ref{eq:es1lg})
\begin{equation}
\left| K_{\mu \gamma }^{\Sigma } \right|\le {\mathcal{C}_{3}}\mathcal{J}_{0}^{2}\left( 1+t \right){{\left( \int\limits_{0}^{{{r}_{0}}}{dr}\left\| D\phi  \right\|_{{{L}^{\infty }}}^{2} \right)}^{1/2}}{{\left\| {{\partial }_{k}}{{g}_{a\bar{b}}} \right\|}_{{{L}^{\infty }}}}{{\left\| \phi  \right\|}_{{{L}^{\infty }}}} ~ ,
\end{equation}
\begin{equation}
\begin{split}
\left| L_{\mu \gamma }^{\Sigma } \right|&\le {\mathcal{C}_{4}}{{\mathcal{J}}_{0}}\left( 1+t \right){{\left( \int\limits_{0}^{{{r}_{0}}}{dr}\left\| \phi  \right\|_{{{L}^{\infty }}}^{2} \right)}^{1/2}}\left\{ \left( \sum\limits_{n=1}^{N}{\left\| \phi  \right\|_{{{L}^{\infty }}}^{n+5}}+\sum\limits_{n=1}^{N}{\left\| \phi  \right\|_{{{L}^{\infty }}}^{n+4}}+\sum\limits_{n=1}^{N}{\left\| \phi  \right\|_{{{L}^{\infty }}}^{n+3}}+\sum\limits_{n=1}^{N}{{{\left| \phi  \right|}^{n+2}}} \right. \right. \\ 
& \left. +\left\| \phi  \right\|_{{{L}^{\infty }}}^{2}+{{\left\| \phi  \right\|}_{{{L}^{\infty }}}}+1 \right)\times \left( 1+{{\left\| \partial \phi  \right\|}_{{{L}^{\infty }}}} \right)\left( {{\left\| \phi  \right\|}_{{{L}^{\infty }}}}+\left\| \phi  \right\|_{{{L}^{\infty }}}^{3}+\sum\limits_{n=1}^{N}{\left\| \phi  \right\|_{{{L}^{\infty }}}^{n+2}}+\sum\limits_{n=1}^{N}{\left\| \phi  \right\|_{{{L}^{\infty }}}^{n+4}} \right) \\ 
& +{{\left\| \phi  \right\|}_{{{L}^{\infty }}}}{{\left\| \partial \phi  \right\|}_{{{L}^{\infty }}}}+{{c}_{4}}\mathcal{I}\left( \left\| \phi  \right\| \right)+{{\left\| \phi  \right\|}_{{{L}^{\infty }}}}\left. +{{\left\| A \right\|}_{{{L}^{\infty }}}} \right\} \\ 
& +{{C}_{4}}{{\mathcal{J}}_{0}}\left( 1+t \right){{\left( \int\limits_{0}^{{{r}_{0}}}{dr}\left\| D\phi  \right\|_{{{L}^{\infty }}}^{2} \right)}^{1/2}}\left( \sum\limits_{n=1}^{N}{\left\| \phi  \right\|_{{{L}^{\infty }}}^{n+5}}+\sum\limits_{n=1}^{N}{\left\| \phi  \right\|_{{{L}^{\infty }}}^{n+4}}+\sum\limits_{n=1}^{N}{\left\| \phi  \right\|_{{{L}^{\infty }}}^{n+3}} \right. \\ 
& \left. +\sum\limits_{n=1}^{N}{{{\left| \phi  \right|}^{n+2}}+\left\| \phi  \right\|_{{{L}^{\infty }}}^{2}+{{\left\| \phi  \right\|}_{{{L}^{\infty }}}}+1} \right) ~ ,  
\end{split}
\end{equation}
respectively, with ${\mathcal{C}_{3}}, {\mathcal{C}_{4}} > 1$. It is important to notice the estimate of the fifth term $ M_{\mu \gamma }^{\Sigma }$ gives the same result as  the estimate of the first term $ I_{\mu \gamma }^{\Sigma }$. 

Next, we need to estimate the linear term of (\ref{es1}). The linear term can be estimated using the initial data
\begin{equation}
\left| {\mathcal{F}_{\mu \gamma }}^{\Sigma |lin} \right|\le \frac{1}{4\pi }{{\int\limits_{{{S}^{2}}}{d\Omega \left| {{r}_{0}}\frac{\partial \left\{ \mathcal{F}_{\mu \gamma }^{\Sigma } \right\}}{\partial t}+{{r}_{0}}\frac{\partial \left\{ \mathcal{F}_{\mu \gamma }^{\Sigma } \right\}}{\partial r}+F_{\mu \gamma }^{\Sigma } \right|}}_{t={{t}_{0}},r={{r}_{0}}}} \le {\tilde{C}_{0}}+{\tilde{C}_{1}}{{r}_{0}} ~ ,
\end{equation}
with $\tilde{C}_0, \tilde{C}_1 \ge 0$. Therefore, the estimate for the gauge field is given by
\begin{equation}\label{result es1}
\begin{split}
\left| \mathcal{F}_{\mu \gamma }^{\Sigma } \right| &\le {\tilde{C}_{0}}+{\tilde{C}_{1}}{{r}_{0}}+\mathcal{J}_{0}^{2}\left( 1+t \right)\left\{ {{\left( \int\limits_{0}^{{{r}_{0}}}{dr}\left\| \mathcal{F}(-r) \right\|_{{{L}^{\infty }}}^{2} \right)}^{1/2}}{\tilde{\mathcal{C}}_{1}} L\left( t \right) \right. \\ 
& \left. +{{\left( \int\limits_{0}^{{{r}_{0}}}{dr}\left\| D\phi  \right\|_{{{L}^{\infty }}}^{2} \right)}^{1/2}}{\tilde{\mathcal{C}}_{2}}M(t)+{{\left( \int\limits_{0}^{{{r}_{0}}}{dr}\left\| \phi  \right\|_{{{L}^{\infty }}}^{2} \right)}^{1/2}}{\tilde{\mathcal{C}}_{3}}N(t) \right\}  ~ ,
\end{split}
\end{equation}
where ${\tilde{\mathcal{C}}_{i }}  > 1$ for all $i=1,2,3$, and
\begin{eqnarray}
L(t) & = & {{\left\| \phi  \right\|}_{{{L}^{\infty }}}}\left( {{\left\| \phi  \right\|}_{{{L}^{\infty }}}}+\left\| \phi  \right\|_{{{L}^{\infty }}}^{3}+\sum\limits_{n=1}^{N}{\left\| \phi  \right\|_{{{L}^{\infty }}}^{n+2}}+\sum\limits_{n=1}^{N}{\left\| \phi  \right\|_{{{L}^{\infty }}}^{n+4}} \right)\left( {{\left\| \partial \phi  \right\|}_{{{L}^{\infty }}}}+1 \right) \nonumber\\ 
& & +{{\left\| \partial \Psi  \right\|}_{{{L}^{\infty }}}}{{\left\| \phi  \right\|}_{{{L}^{\infty }}}}+{{\left\| \partial \phi  \right\|}_{{{L}^{\infty }}}}+\left\| \phi  \right\|_{{{L}^{\infty }}}^{2}{{\left\| \partial \phi  \right\|}_{{{L}^{\infty }}}}+{{\left\| \phi  \right\|}_{{{L}^{\infty }}}}\mathcal{I}\left( \left\| \phi  \right\| \right) ~ ,\\
M(t) &=& {{\left\| \phi  \right\|}_{{{L}^{\infty }}}}{{\left\| \partial \Psi  \right\|}_{{{L}^{\infty }}}}+\left\| \phi  \right\|_{{{L}^{\infty }}}^{2}+\sum\limits_{n=1}^{N}{\frac{\left( n+2 \right)}{n+1}{{b}_{n}}\left\| \phi  \right\|_{{{L}^{\infty }}}^{n+3}}+{{C}_{1}}\left\| \phi  \right\|_{{{L}^{\infty }}}^{2}+\left\| \phi  \right\|_{{{L}^{\infty }}}^{2}\nonumber\\ 
& &+{{\left\| \phi  \right\|}_{{{L}^{\infty }}}}+\sum\limits_{n=1}^{N}{\left\| \phi  \right\|_{{{L}^{\infty }}}^{n+5}}+\sum\limits_{n=1}^{N}{\left\| \phi  \right\|_{{{L}^{\infty }}}^{n+4}}+\sum\limits_{n=1}^{N}{\left\| \phi  \right\|_{{{L}^{\infty }}}^{n+3}}+\sum\limits_{n=1}^{N}{{{\left| \phi  \right|}^{n+2}}}+1  ~ ,\\  
N(t) &=& \sum\limits_{n=1}^{N}{\left\| \phi  \right\|_{{{L}^{\infty }}}^{n+5}}+\sum\limits_{n=1}^{N}{\left\| \phi  \right\|_{{{L}^{\infty }}}^{n+4}}+\sum\limits_{n=1}^{N}{\left\| \phi  \right\|_{{{L}^{\infty }}}^{n+3}}+\sum\limits_{n=1}^{N}{{{\left| \phi  \right|}^{n+2}}}+\left\| \phi  \right\|_{{{L}^{\infty }}}^{2}+{{\left\| \phi  \right\|}_{{{L}^{\infty }}}}+1 \nonumber\\ 
&& \left( {{\left\| \phi  \right\|}_{{{L}^{\infty }}}}+\left\| \phi  \right\|_{{{L}^{\infty }}}^{3}+\sum\limits_{n=1}^{N}{\left\| \phi  \right\|_{{{L}^{\infty }}}^{n+2}}+\sum\limits_{n=1}^{N}{\left\| \phi  \right\|_{{{L}^{\infty }}}^{n+4}} \right)\left( {{\left\| \partial \phi  \right\|}_{{{L}^{\infty }}}}+1 \right) \nonumber\\ 
&& +{{\left\| \phi  \right\|}_{{{L}^{\infty }}}}{{\left\| \partial \phi  \right\|}_{{{L}^{\infty }}}}+{{c}_{4}}\mathcal{I}\left( \left\| \phi  \right\| \right) ~  .
\end{eqnarray}

\subsection{Estimate for the complex scalar fields}

Let us rewrite the integral equation (\ref{es2}) as
\begin{equation}
{{D}_{\mu }}{{\phi }^{a}} = {{D}_{\mu }}{{\phi }^{a}}^{|lin} + N_{\mu}^{a} + R_{\mu}^{a} ~ ,
\end{equation}
where $N_{\mu}^{a}$  and $R_{\mu}^{a}$  are the nonlinear part of (\ref{es2}) whose forms are defined as 
\begin{equation}\label{eq:nonlinN}
N_{\mu}^{a}  \equiv  \frac{i}{4\pi }\int\limits_{{{K}_{p}}}{rdrd\Omega } ~  q_{\Gamma } \left. \left(  \partial^{\alpha } {\mathcal{F}}_{\mu \alpha }^{\Gamma }{{\phi }^{a}}+ { \mathcal{F}}_{\mu \alpha }^{\Gamma }\partial^{\alpha }{\phi }^{a} + \partial^{\alpha } \left( A_{\alpha }^{\Gamma }\partial_{\mu } {\phi }^{a} \right) -  \partial^{\alpha } \left( A_{\mu }^{\Gamma } \partial_{\alpha } {\phi }^{a}  \right) \right) \right|_{t=-r}    ~ ,
\end{equation}
\begin{equation}\label{eq:nonlinR}
R_{\mu}^{a}  \equiv   \frac{1}{4\pi} \int\limits_{{{K}_{p}}}{rdrd\Omega } \left.  {{\partial }_{\mu }}{{\partial }^{\alpha }}{{D}_{\alpha }}{{\phi }^{a}}  \right|_{t=-r}  ~ ,
\end{equation}

By applying Holder inequality, Sobolev estimates, and the result in (\ref{es eom 1}), the  estimate of   $N_{\mu}^{a}$ has the form
\begin{equation}
\begin{split}
\left| N_{\mu }^{a} \right|&\le {{K}_{1}}{{\mathcal{J}}_{0}}{{\left( \int\limits_{0}^{{{r}_{0}}}{dr}\left\| \phi  \right\|_{{{L}^{\infty }}}^{2} \right)}^{1/2}}{{\left\| \phi  \right\|}_{{{L}^{\infty }}}}\left( 1+{{\left\| \partial \phi  \right\|}_{{{L}^{\infty }}}} \right) \\ 
& +{{K}_{2}}\mathcal{J}_{0}^{2}\left( 1+t \right){{\left( \int\limits_{0}^{{{r}_{0}}}{dr}\left\| ^{(4)}F(-r) \right\|_{{{L}^{\infty }}}^{2} \right)}^{1/2}} \\ 
& +{{K}_{3}}{{\mathcal{J}}_{0}}{{\left( \int\limits_{0}^{{{r}_{0}}}{dr}\left\| A \right\|_{{{L}^{\infty }}}^{2} \right)}^{1/2}}\Bigg\{ \left( {{\left\| \phi  \right\|}_{{{L}^{\infty }}}}+\left\| \phi  \right\|_{{{L}^{\infty }}}^{3}+\sum\limits_{n=1}^{N}{\left\| \phi  \right\|_{{{L}^{\infty }}}^{n+2}}+\sum\limits_{n=1}^{N}{\left\| \phi  \right\|_{{{L}^{\infty }}}^{n+4}} \right)\left( 1+{{\left\| \phi  \right\|}_{{{L}^{\infty }}}} \right)  \\ 
&  +{{\left\| \phi  \right\|}_{{{L}^{\infty }}}}{{\left\| \partial \phi  \right\|}_{{{L}^{\infty }}}}+{{c}_{4}}\mathcal{I}\left( \left\| \phi  \right\| \right) \Bigg\}+{{K}_{3}}{{\mathcal{J}}_{0}}{{\left( \int\limits_{0}^{{{r}_{0}}}{dr}\left\| \partial \phi  \right\|_{{{L}^{\infty }}}^{2} \right)}^{1/2}} ~ ,\\ 
\end{split}
\end{equation}
where $K_i > 1$ for all $i=1,2,3$. As for $R_{\mu}^{a}$ , we first consider
\begin{equation}
\left| R_{\mu }^{a} \right|\le  \int\limits_{{{K}_{p}}}{rdrd\Omega }{{\left| {{\partial }_{\mu }}\left( \overline{{{D}^{\alpha }}}{{D}_{\alpha }}{{\phi }^{a}} \right)-i{{q}_{\Gamma }}{{\partial }_{\mu }}\left( {{A}^{\Gamma |\alpha }}{{D}_{\alpha }}{{\phi }^{a}} \right) \right|}_{t=-r}}  ~ .
\end{equation}
To estimate this term, we must derive an estimate for the first derivative of the equation (\ref{eom2}), that is, 
\begin{equation}
\begin{split}
{{\partial }_{\mu }}\left( \overline{{{D}^{\alpha }}}{{D}_{\alpha }}{{\phi }^{a}} \right)&={{g}^{d\bar{b}}}\left( \frac{1}{4}{{\partial }_{\mu }}\mathcal{F}_{\alpha \beta }^{\Lambda }{{\partial }_{d}}G_{\Lambda }^{\alpha \beta }+\frac{1}{4}\mathcal{F}_{\alpha \beta }^{\Lambda }{{\partial }_{\mu }}{{\partial }_{d}}G_{\Lambda }^{\alpha \beta } \right.-{{\partial }_{\mu }}{{\partial }_{d}}{{g}_{a\bar{c}}}{{D}_{\alpha }}{{\phi }^{a}}\overline{{{D}^{\alpha }}{{\phi }^{c}}} \\ 
& -{{\partial }_{d}}{{g}_{a\bar{c}}}{{\partial }_{\mu }}{{D}_{\alpha }}{{\phi }^{a}}\overline{{{D}^{\alpha }}{{\phi }^{c}}}-{{\partial }_{d}}{{g}_{a\bar{c}}}{{D}_{\alpha }}{{\phi }^{a}}{{\partial }_{\mu }}\overline{{{D}^{\alpha }}{{\phi }^{c}}}-{{\partial }_{\mu }}{{\partial }^{\alpha }}{{g}_{d\bar{c}}}\overline{{{D}_{\alpha }}{{\phi }^{c}}} \\ 
& \left. -{{\partial }^{\alpha }}{{g}_{d\bar{c}}}{{\partial }_{\mu }}\overline{{{D}_{\alpha }}{{\phi }^{c}}}-{{\partial }_{\mu }}{{\partial }_{d}}V \right)+{{\partial }_{\mu }}{{g}^{d\bar{b}}}{{g}_{d\bar{b}}}\left( \overline{{{D}^{\alpha }}}{{D}_{\alpha }}{{\phi }^{a}} \right) ~ , \\  
\end{split}
\end{equation}
which can be bound by an expression involving the energy and $L^\infty$ norm. Hence, the estimate  of   $R_{\mu}^{a}$ is given by
\begin{equation}
\begin{split}
\left| R_{\mu }^{a} \right| &\le {{\mathcal{B}}_{1}}\mathcal{J}{{}_{0}}{{\left( \int\limits_{0}^{{{r}_{0}}}{dr}\left\| \partial \phi  \right\|_{{{L}^{\infty }}}^{2} \right)}^{1/2}}{{\left\{ \left( {{\left\| \phi  \right\|}_{{{L}^{\infty }}}}+\left\| \phi  \right\|_{{{L}^{\infty }}}^{3}+\sum\limits_{n=1}^{N}{\left\| \phi  \right\|_{{{L}^{\infty }}}^{n+2}}+\sum\limits_{n=1}^{N}{\left\| \phi  \right\|_{{{L}^{\infty }}}^{n+4}} \right) \right.}^{2}} \\ 
& \left. \times \left( 1+{{\left\| \partial \phi  \right\|}_{{{L}^{\infty }}}} \right)+{{\left\| \phi  \right\|}_{{{L}^{\infty }}}}{{\left\| \partial \phi  \right\|}_{{{L}^{\infty }}}}+{{c}_{4}}\mathcal{I}\left( \left\| \phi  \right\| \right) \right\} \\ 
& +{{\mathcal{B}}_{2}}{{\mathcal{J}}_{0}}{{\left( \int\limits_{0}^{{{r}_{0}}}{dr}\left\| ^{(4)}F(-r) \right\|_{{{L}^{\infty }}}^{2} \right)}^{1/2}}\left( \left\| \partial \phi  \right\|_{{{L}^{\infty }}}^{2}\left\| \phi \text{ } \right\|_{{{L}^{\infty }}}^{2}+{{\left\| \partial \phi  \right\|}_{{{L}^{\infty }}}}\left\| \phi \text{ } \right\|_{{{L}^{\infty }}}^{2} \right. \\ 
& \left. +{{\left\| \phi  \right\|}_{{{L}^{\infty }}}}+{{\left\| \partial \phi  \right\|}_{{{L}^{\infty }}}}+\left\| \phi \text{ } \right\|_{{{L}^{\infty }}}^{2} \right) \\ 
& +{{\mathcal{B}}_{3}}{{\mathcal{J}}_{0}}{{\left( \int\limits_{0}^{{{r}_{0}}}{dr}\left\| D\phi (-r) \right\|_{{{L}^{\infty }}}^{2} \right)}^{1/2}}\left\{ {{\left\| \partial \phi  \right\|}_{{{L}^{\infty }}}}\left( \sum\limits_{n=1}^{N}{\left\| \phi  \right\|_{{{L}^{\infty }}}^{n+2}}+\sum\limits_{n=1}^{N}{\left\| \phi  \right\|_{{{L}^{\infty }}}^{n+1}}+1 \right) \right. \\ 
& +{{\left\| \partial \phi  \right\|}_{{{L}^{\infty }}}}\left( {{\left\| \phi  \right\|}_{{{L}^{\infty }}}}+\left\| \phi  \right\|_{{{L}^{\infty }}}^{3}+\sum\limits_{n=1}^{N}{\left\| \phi  \right\|_{{{L}^{\infty }}}^{n+2}}+\sum\limits_{n=1}^{N}{\left\| \phi  \right\|_{{{L}^{\infty }}}^{n+4}} \right)\left( 1+{{\left\| \phi  \right\|}_{{{L}^{\infty }}}} \right) \\ 
& \left. +\mathcal{I}\left( \left\| \phi  \right\| \right)+{{\left\| \phi  \right\|}_{{{L}^{\infty }}}}+\left( 1+t \right){{\left\| A \right\|}_{{{L}^{\infty }}}} \right. \Bigg\} \\ 
& +{{\mathcal{B}}_{4}}{{\left( \int\limits_{0}^{{{r}_{0}}}{dr}\left\| \partial \phi  \right\|_{{{L}^{\infty }}}^{2} \right)}^{1/2}}\mathcal{H}\left( \left\| \phi  \right\| \right)\,
\end{split}
\end{equation}
where ${\mathcal B}_{\hat{i}} > 1$ for all $\hat{i} = 1, ...,4$,
\begin{equation}
\mathcal{H}\left( \left\| \phi  \right\| \right)\le\begin{cases}
\mathcal{D}\left( \left\| \phi  \right\| \right) ~, &\text{ if $V$ is of the form \eqref{potential1}} ~ ,\\
{{\mathcal{J}}_{0}}\left( \left\| \partial \Psi  \right\|_{{{L}^{\infty }}}^{2}+1 \right) ~ , & \text{ if $V$  is either of the form \eqref{potential2} or \eqref{potential3} } ~,
\end{cases}
\end{equation}
and
\begin{equation}
\mathcal{D}\left( {{\left\| \phi  \right\|}_{{{L}^{\infty }}}} \right)\le\sum\limits_{n=0}^{N-1}{\left\| \phi  \right\|_{{{L}^{\infty }}}^{2n}}+{{\mathcal{J}}_{0}}\left( 1+t \right){{\left\| {{\partial }_{c}}\Psi  \right\|}_{{{L}^{\infty }}}}\sum\limits_{n=0}^{N-2}{\left\| \phi  \right\|_{{{L}^{\infty }}}^{2n}}\:.
\end{equation}
In the same way as the gauge field case, the linear term of (\ref{es2}) is bounded by the initial data. So, the total estimate for the complex scalar field is given by
\begin{equation}\label{result es2}
\begin{split}
\left| {{D}_{\mu }}{{\phi }^{a}} \right|&\le {{\mathcal{K}}_{1}}{{\mathcal{J}}_{0}}{{\left( \int\limits_{0}^{{{r}_{0}}}{dr}\left\| D\phi (-r) \right\|_{{{L}^{\infty }}}^{2} \right)}^{1/2}}\mathcal{S}\left( t \right)+{{\mathcal{K}}_{2}}\mathcal{J}_{0}^{2}\left( 1+t \right){{\left( \int\limits_{0}^{{{r}_{0}}}{dr}\left\| ^{(4)}\mathcal{F}(-r) \right\|_{{{L}^{\infty }}}^{2} \right)}^{1/2}}\mathcal{X}\left( t \right) \\ 
& +{{\mathcal{K}}_{3}}{{\mathcal{J}}_{0}}{{\left( \int\limits_{0}^{{{r}_{0}}}{dr}\left\| \phi  \right\|_{{{L}^{\infty }}}^{2} \right)}^{1/2}}{{\left\| \phi  \right\|}_{{{L}^{\infty }}}}\left( 1+{{\left\| \partial \phi  \right\|}_{{{L}^{\infty }}}} \right)+{{\mathcal{K}}_{4}}{{\mathcal{J}}_{0}}{{\left( \int\limits_{0}^{{{r}_{0}}}{dr}\left\| A \right\|_{{{L}^{\infty }}}^{2} \right)}^{1/2}}\mathcal{U}\left( t \right) \\ 
& +{{\mathcal{K}}_{5}}\mathcal{J}{{}_{0}}{{\left( \int\limits_{0}^{{{r}_{0}}}{dr}\left\| \partial \phi  \right\|_{{{L}^{\infty }}}^{2} \right)}^{1/2}}\mathcal{W}\left( t \right)+{{k}_{0}}+{{k}_{1}}{{r}_{0}}  
\end{split}
\end{equation}
where ${\mathcal K}_{\tilde{i}} > 1$ for all $\tilde{i} =1,...,5$, and 
\begin{eqnarray}
\mathcal{S}\left( t \right)&=&{{\left\| \partial \phi  \right\|}_{{{L}^{\infty }}}}\left( \sum\limits_{n=1}^{N}{\left\| \phi  \right\|_{{{L}^{\infty }}}^{n+2}}+\sum\limits_{n=1}^{N}{\left\| \phi  \right\|_{{{L}^{\infty }}}^{n+1}}+1 \right)+\mathcal{I}\left( \left\| \phi  \right\| \right)+{{\left\| \phi  \right\|}_{{{L}^{\infty }}}}+\left( 1+t \right){{\left\| A \right\|}_{{{L}^{\infty }}}} \nonumber\\ 
& &\quad+{{\left\| \partial \phi  \right\|}_{{{L}^{\infty }}}}\left( {{\left\| \phi  \right\|}_{{{L}^{\infty }}}}+\left\| \phi  \right\|_{{{L}^{\infty }}}^{3}+\sum\limits_{n=1}^{N}{\left\| \phi  \right\|_{{{L}^{\infty }}}^{n+2}}+\sum\limits_{n=1}^{N}{\left\| \phi  \right\|_{{{L}^{\infty }}}^{n+4}} \right)\left( 1+{{\left\| \phi  \right\|}_{{{L}^{\infty }}}} \right) ~ ,\\
\mathcal{X}\left( t \right)& = &1+\left\| \partial \phi  \right\|_{{{L}^{\infty }}}^{2}\left\| \phi \text{ } \right\|_{{{L}^{\infty }}}^{2}+{{\left\| \partial \phi  \right\|}_{{{L}^{\infty }}}}\left\| \phi \text{ } \right\|_{{{L}^{\infty }}}^{2}+{{\left\| \phi  \right\|}_{{{L}^{\infty }}}}+{{\left\| \partial \phi  \right\|}_{{{L}^{\infty }}}}+\left\| \phi \text{ } \right\|_{{{L}^{\infty }}}^{2} ~ ,\\
\mathcal{U}\left( t \right )& = &\left( {{\left\| \phi  \right\|}_{{{L}^{\infty }}}}+\left\| \phi  \right\|_{{{L}^{\infty }}}^{3}+\sum\limits_{n=1}^{N}{\left\| \phi  \right\|_{{{L}^{\infty }}}^{n+2}}+\sum\limits_{n=1}^{N}{\left\| \phi  \right\|_{{{L}^{\infty }}}^{n+4}} \right)\left( 1+{{\left\| \phi  \right\|}_{{{L}^{\infty }}}} \right) \nonumber\\
& & \quad + {{\left\| \phi  \right\|}_{{{L}^{\infty }}}}{{\left\| \partial \phi  \right\|}_{{{L}^{\infty }}}}+{{c}_{4}}\mathcal{I}\left( \left\| \phi  \right\| \right) ~ ,\\
\mathcal{W}\left( t \right)&=&\left\{ \left( {{\left\| \phi  \right\|}_{{{L}^{\infty }}}}+\left\| \phi  \right\|_{{{L}^{\infty }}}^{3}+\sum\limits_{n=1}^{N}{\left\| \phi  \right\|_{{{L}^{\infty }}}^{n+2}}+\sum\limits_{n=1}^{N}{\left\| \phi  \right\|_{{{L}^{\infty }}}^{n+4}} \right) \right.\left. \left( 1+{{\left\| \partial \phi  \right\|}_{{{L}^{\infty }}}} \right)+{{\left\| \phi  \right\|}_{{{L}^{\infty }}}}{{\left\| \partial \phi  \right\|}_{{{L}^{\infty }}}} \right. \nonumber\\ 
& & \left.+{{c}_{4}}\mathcal{I}\left( \left\| \phi  \right\| \right)\right. \Bigg\}  \left( {{\left\| \phi  \right\|}_{{{L}^{\infty }}}}+\left\| \phi  \right\|_{{{L}^{\infty }}}^{3}+\sum\limits_{n=1}^{N}{\left\| \phi  \right\|_{{{L}^{\infty }}}^{n+2}}+\sum\limits_{n=1}^{N}{\left\| \phi  \right\|_{{{L}^{\infty }}}^{n+4}} \right)+\mathcal{H}\left( {{\left\| \phi  \right\|}_{{{L}^{\infty }}}} \right) ~ . \nonumber\\
\end{eqnarray}

\section{The Global Existence}
\label{sec:GlobalExistence}
In this section we put the final argument and prove the global existence of the MKG system with general coupling. 

Since the right hand side of \ref{result es1}) and (\ref{result es2}) independent of the spatial coordinates of $p$ and reversing the steps which shifted the origin of coordinates, we can write the result as
\begin{equation}
\begin{split}
\left\| ^{(4)}\mathcal{F}(t) \right\|_{{{L}^{\infty }}}&\le \mathcal{J}_{0}^{2}\left( 1+t \right)\left\{ {{\mathcal{C}}_{1}}L(t){{\left( \int\limits_{0}^{t}{ds}\left\| \mathcal{F}(s) \right\|_{{{L}^{\infty }}}^{2} \right)}^{1/2}}+{{\mathcal{C}}_{2}}M(t){{\left( \int\limits_{0}^{t}{ds}\left\| D\phi (s) \right\|_{{{L}^{\infty }}}^{2} \right)}^{1/2}} \right. \\ 
& \left. +{{\mathcal{C}}_{3}}N(t){{\left( \int\limits_{0}^{t}{ds}\left\| \phi (s) \right\|_{{{L}^{\infty }}}^{2} \right)}^{1/2}} \right\}+{{c}_{0}}+{{c}_{1}}t  ~ ,
\end{split}
\end{equation}
\begin{equation}
\begin{split}
\left\| D\phi  \right\|_{{{L}^{\infty }}}&\le {{\mathcal{K}}_{1}}{{\mathcal{J}}_{0}}{{\left( \int\limits_{0}^{{{r}_{0}}}{dr}\left\| D\phi (-r) \right\|_{{{L}^{\infty }}}^{2} \right)}^{1/2}}\mathcal{S}\left( t \right)+{{\mathcal{K}}_{2}}\mathcal{J}_{0}^{2}\left( 1+t \right){{\left( \int\limits_{0}^{t}{ds}\left\| ^{(4)}\mathcal{F}(s) \right\|_{{{L}^{\infty }}}^{2} \right)}^{1/2}}\mathcal{X}\left( t \right) \\ 
& +{{\mathcal{K}}_{3}}{{\mathcal{J}}_{0}}{{\left( \int\limits_{0}^{t}{ds}\left\| \phi (s) \right\|_{{{L}^{\infty }}}^{2} \right)}^{1/2}}{{\left\| \phi  \right\|}_{{{L}^{\infty }}}}\left( 1+{{\left\| \partial \phi  \right\|}_{{{L}^{\infty }}}} \right)+{{\mathcal{K}}_{4}}{{\mathcal{J}}_{0}}{{\left( \int\limits_{0}^{t}{ds}\left\| A(s) \right\|_{{{L}^{\infty }}}^{2} \right)}^{1/2}}\mathcal{U}\left( t \right) \\ 
& +{{\mathcal{K}}_{5}}\mathcal{J}{{}_{0}}{{\left( \int\limits_{0}^{t}{ds}\left\| \partial \phi (s) \right\|_{{{L}^{\infty }}}^{2} \right)}^{1/2}}\mathcal{W}\left( t \right)+{{k}_{0}}+{{k}_{1}}t ~ .
\end{split}
\end{equation}

Furthermore, to prove the global existence of the MKG equation, we must show that the norm $\left\| D\phi \left( t \right) \right\|_{{{L}^{\infty }}}^{{}}$ dan $\left\| ^{(4)}\mathcal{F}(t) \right\|_{{{L}^{\infty }}}^{{}}$ is finite. Therefore, we define a function as follows
\begin{equation}
\mathcal{G}\left( t \right)=\left\| ^{(4)}\mathcal{F}(t) \right\|_{{{L}^{\infty }}}^{{}}+\left\| D\phi \left( t \right) \right\|_{{{L}^{\infty }}} ~ .
\end{equation}
Restating the equation in the form of Gronwall inequality, we obtain
\begin{equation}
{{\mathcal{G}}^{2}}\left( t \right)\le \mathcal{N}\left( t \right)+\mathcal{Q}\left( t \right)\left( \int\limits_{0}^{t}{ds}\left\{ \left\| \mathcal{F}(s) \right\|_{{{L}^{\infty }}}^{2}+\left\| D\phi (s) \right\|_{{{L}^{\infty }}}^{2}+\left\| \phi (s) \right\|_{{{L}^{\infty }}}^{2}+\left\| \partial \phi (s) \right\|_{{{L}^{\infty }}}^{2}+\left\| A(s) \right\|_{{{L}^{\infty }}}^{2} \right\} \right) ~ ,
\end{equation}
with
\begin{eqnarray*}
\mathcal{Q}\left( t \right) &=&{{\mathcal{K}}_{1}}{{\mathcal{J}}_{0}}\mathcal{S}\left( t \right)+{{\mathcal{K}}_{2}}\mathcal{J}_{0}^{2}\left( 1+t \right)\mathcal{X}\left( t \right)+{{\mathcal{K}}_{3}}{{\mathcal{J}}_{0}}{{\left\| \phi  \right\|}_{{{L}^{\infty }}}}\left( 1+{{\left\| \partial \phi  \right\|}_{{{L}^{\infty }}}} \right) \nonumber\\ 
& &+{{\mathcal{K}}_{4}}{{\mathcal{J}}_{0}}\mathcal{U}\left( t \right)+{{\mathcal{K}}_{5}}\mathcal{J}{{}_{0}}\mathcal{W}\left( t \right)+\mathcal{J}_{0}^{2}\left( 1+t \right)\left\{ {{\mathcal{C}}_{1}}L(t)+{{\mathcal{C}}_{2}}M(t) \right.\left. +{{\mathcal{C}}_{3}}N(t) \right\} ~ , \\
\mathcal{N}\left( t \right)& = & {{c}_{0}}+{{c}_{1}}t+{{k}_{0}}+{{k}_{1}}t ~ .
\end{eqnarray*}

To get a bound of ${{\mathcal{G}}^{2}}\left( t \right)$, we only need to prove that ${{\mathcal{G}}^{2}}\left( t \right)$ continuous. Continuity of ${{\mathcal{G}}^{2}}\left( t \right)$ is depend on continuity of $\left\| ^{(4)}\mathcal{F}(t) \right\|_{{{L}^{\infty }}}^{{}},\left\| \phi (t) \right\|_{{{L}^{\infty }}}^{{}},\left\| \partial \phi (t) \right\|_{{{L}^{\infty }}}^{{}},\left\| A(t) \right\|_{{{L}^{\infty }}}^{{}}$ and $\left\| D\phi (t) \right\|_{{{L}^{\infty }}}$. Using the triangle inequality and the Sobolev estimate
\begin{equation}
\left\| f \right\|_{{{L}^{\infty }}}^{{}}\le \left\| f \right\|_{{{H}_{2}}} ~ ,
\end{equation}  
we can get the continuity of $\left\| ^{(4)}\mathcal{F}(t) \right\|_{{{L}^{\infty }}}$. 
Let $\varepsilon >0$, then we get
\begin{equation}
\begin{split}
\left| \left\| ^{(4)}\mathcal{F}(t+\varepsilon ) \right\|_{{{L}^{\infty }}}^{{}}-\left\| ^{(4)}\mathcal{F}(t) \right\|_{{{L}^{\infty }}}^{{}} \right|&\le \left\| ^{(4)}\mathcal{F}(t+\varepsilon ){{-}^{(4)}}F(t) \right\|_{{{L}^{\infty }}}^{{}} \\ 
&\le \left\| ^{(4)}\mathcal{F}(t+\varepsilon ){{-}^{(4)}}\mathcal{F}(t) \right\|_{{{H}_{2}}}^{{}} ~ , \\ 
\end{split}
\end{equation}
hence, we have
\begin{equation}
\left\| ^{(4)}\mathcal{F}(t+\varepsilon ){{-}^{(4)}}\mathcal{F}(t) \right\|_{{{H}_{2}}}^{{}}\to 0 ~ , \quad \varepsilon \to \text{0} ~ .
\end{equation}
The last step follows from continuity of $^{(4)}\mathcal{F}(t)$ as a curve in $H_2$. The same reason clearly applies to $\left\| D\phi (t) \right\|_{{{L}^{\infty }}}$.

So far we have proved that $\left\| ^{(4)}\mathcal{F}(t) \right\|_{{{L}^{\infty }}}$ and $\left\| D\phi (t) \right\|_{{{L}^{\infty }}}$ cannot blow up in a finite time. Another estimate we need to prove are $\left\| {{A}_{i}}\left( t \right) \right\|_{{{L}^{\infty }}}^{{}}$ and $\left\| \phi \left( t \right) \right\|_{{{L}^{\infty }}}^{{}}$. By considering the temporal gauge condition, we have
\begin{equation}
{{A}_{i}}\left( t,x \right)={{A}_{i}}\left( 0,x \right)+\int\limits_{0}^{t}{{{E}_{i}}\left( s \right)ds}\;,
\end{equation}
then we get
\begin{equation}
\left\| {{A}_{i}}\left( t,x \right) \right\|_{{{L}^{\infty }}}^{{}}=\left\| {{A}_{i}}\left( 0,x \right) \right\|_{{{L}^{\infty }}}^{{}}+\int\limits_{0}^{t}{\left\| {{E}_{i}}\left( s \right) \right\|_{{{L}^{\infty }}}^{{}}ds}\;,
\end{equation}
\begin{equation}
\left\| \phi \left( t,x \right) \right\|_{{{L}^{\infty }}}^{{}}=\left\| \phi \left( 0,x \right) \right\|_{{{L}^{\infty }}}^{{}}+\int\limits_{0}^{t}{\left\| {{\partial }_{0}}\phi \left( s \right) \right\|_{{{L}^{\infty }}}^{{}}ds}\:,
\end{equation}
which are the key points to complete the proof of the global existence of Maxwell Klein-Gordon system.

To prove the global existence, we must show that the norm ${{\left( {{H}_{2}}\times {{H}_{1}} \right)}^{2}}$ of $\left( {{A}_{i}},{{E}_{i}},\phi ,{{\partial }_{t}}\phi  \right)$ does not blow up for a finite time. Therefore, we define a functions that are elements of ${{\left( {{H}_{2}}\times {{H}_{1}} \right)}^{2}}$ as
\begin{equation}
{{\calE}_{0}}\text{= }\frac{1}{2}\int\limits_{{{R}^{3}}}{dx\text{ }}\left( {{\delta }_{\Lambda \Sigma }}\left\{ E_{i}^{\Lambda }E_{i}^{\Sigma }+{{\partial }_{j}}A_{i}^{\Lambda }\text{ }{{\partial }_{j}}A_{i}^{\Sigma }+mA_{i}^{\Lambda }A_{i}^{\Sigma } \right\}+{{\left| {{\partial }_{0}}\phi  \right|}^{2}}+{{\left| {{\partial }_{i}}\phi  \right|}^{2}}+m{{\left| \phi  \right|}^{2}} \right) ~ ,
\end{equation}
\begin{equation}
{{\calE}_{1}}\text{= }\frac{1}{2}\int\limits_{{{R}^{3}}}{dx\left\{ {{\delta }_{\Lambda \Sigma }}\left( {{\partial }_{j}}E_{i}^{\Lambda }\text{ }{{\partial }_{j}}E_{i}^{\Sigma }+{{\partial }_{j}}{{\partial }_{k}}A_{i}^{\Lambda }\text{ }{{\partial }_{j}}{{\partial }_{k}}A_{i}^{\Sigma } \right)+{{\left| {{\partial }_{i}}{{\partial }_{0}}\phi  \right|}^{2}}+{{\left| {{\partial }_{j}}{{\partial }_{i}}\phi  \right|}^{2}} \right\}} ~ ,
\end{equation}
where $m>0$, is a positive constant. The function ${{\left( {{\calE}_{0}}+{{\calE}_{1}} \right)}^{1/2}}$ meets the norm ${{\left( {{H}_{2}}\times {{H}_{1}} \right)}^{2}}$, so that to obtain the global existence of MKG equations with general coupling is sufficient by showing that ${{\calE}_{0}}$ and ${{\calE}_{1}}$ are not blow-up for a finite time.

The first derivative with respect to time of $\calE_{0}$,
\begin{equation}
\begin{split}
\left| \frac{d{{\calE}_{0}}}{dt} \right| &\le {{\mathcal{C}}_{0}}\left\{ Y\left( t \right) \right.\left( \left\| D\phi  \right\|_{{{L}^{\infty }}}^{{}}+1+\left\| \phi  \right\|_{{{L}^{\infty }}}^{{}} \right)+\left\| ^{(4)}F(t) \right\|_{{{L}^{\infty }}}^{{}}\left\| \partial \phi  \right\|_{{{L}^{\infty }}}^{{}}\left( 1+\left\| \phi  \right\|_{{{L}^{\infty }}}^{{}} \right) \\ 
& +\left( {{\left\| \partial \phi  \right\|}_{{{L}^{\infty }}}}+{{\left\| D\phi  \right\|}_{{{L}^{\infty }}}} \right)Z\left( t \right)\left. +\mathcal{I}\left( {{\left\| \phi  \right\|}_{{{L}^{\infty }}}} \right)+1 \right\}
{{\calE}_{0}}  \:,
\end{split}
\end{equation}
with
\begin{equation}
\begin{split}
Y\left( t \right)&=8\sum\limits_{n=1}^{N}{{{b}_{n}}}\left\| \phi  \right\|_{{{L}^{\infty }}}^{n+6}+\sum\limits_{n=1}^{N}{{{b}_{n}}}\left\| \phi  \right\|_{{{L}^{\infty }}}^{n+5}+12\sum\limits_{n=1}^{N}{{{b}_{n}}}\left\| \phi  \right\|_{{{L}^{\infty }}}^{n+3} \\ 
&+6{{C}_{1}}\left( \left\| \phi  \right\|_{{{L}^{\infty }}}^{2}+\left\| \phi  \right\|_{{{L}^{\infty }}}^{3} \right)+\left( {{C}_{2}}+{{C}_{3}} \right)\left\| \phi  \right\|_{{{L}^{\infty }}}^{{}}+{{C}_{3}} \\ 
\end{split}
\end{equation}
\begin{equation}
Z\left( t \right)={{\left\| \phi  \right\|}_{{{L}^{\infty }}}}+\left\| \phi  \right\|_{{{L}^{\infty }}}^{3}+\sum\limits_{n=1}^{N}{\left\| \phi  \right\|_{{{L}^{\infty }}}^{n+2}}+\sum\limits_{n=1}^{N}{\left\| \phi  \right\|_{{{L}^{\infty }}}^{n+4}}
\end{equation}	
where we used the Holder inequality.
Integrating the inequality above, we obtain
\begin{equation}
{{\calE}_{0}}\left( t \right)\le {{\calE}_{0}}\left( 0 \right)\exp \left( \int\limits_{0}^{t}{\mathcal{P}\left( t \right)dt} \right) ~ ,
\end{equation}
with
\begin{equation}
\begin{split}
\mathcal{P}\left( t \right)&=Y\left( t \right)\left( \left\| D\phi  \right\|_{{{L}^{\infty }}}^{{}}+1+\left\| \phi  \right\|_{{{L}^{\infty }}}^{{}} \right)+\left\| ^{(4)}F(t) \right\|_{{{L}^{\infty }}}^{{}}\left\| \partial \phi  \right\|_{{{L}^{\infty }}}^{{}}\left( 1+\left\| \phi  \right\|_{{{L}^{\infty }}}^{{}} \right) \\ 
& +\left( {{\left\| \partial \phi  \right\|}_{{{L}^{\infty }}}}+{{\left\| D\phi  \right\|}_{{{L}^{\infty }}}} \right)Z\left( t \right)+\mathcal{I}\left( {{\left\| \phi  \right\|}_{{{L}^{\infty }}}} \right)+1 ~ .
\end{split}
\end{equation}
Since all of $\left\| ^{(4)}\mathcal{F}(t) \right\|_{{{L}^{\infty }}}^{{}},\left\| \phi (t) \right\|_{{{L}^{\infty }}}^{{}},\left\| \partial \phi (t) \right\|_{{{L}^{\infty }}}^{{}},\left\| A(t) \right\|_{{{L}^{\infty }}}^{{}}$ and $\left\| D\phi (t) \right\|_{{{L}^{\infty }}}^{{}}$ does not blow-up for a finite time, therefore, ${{\calE}_{0}}\left( t \right)$ is bounded for all time.

Finally, computing the time derivative of $\calE_{1}$, and after some calculations, we have
\begin{equation}
\left| \frac{d{{\calE}_{1}}}{dt} \right|\le \left\{{X}\left( t \right)+{W}\left( t \right)+{P}\left( t \right)+{U}\left( t \right) \right\}{{\calE}_{1}} ~ ,
\end{equation}
with
\begin{equation}
\begin{split}
X\left( t \right)&=Y\left( t \right)\left\{ \left( {{\left\| D\phi  \right\|}_{{{L}^{\infty }}}}{{\left\| \partial \phi  \right\|}_{{{L}^{\infty }}}}+\left\| D\phi (t) \right\|_{{{L}^{\infty }}}^{{}}+\left\| \phi  \right\|_{{{L}^{\infty }}}^{{}}+\left\| \partial \phi  \right\|_{{{L}^{\infty }}}^{{}}+1 \right)\calE_{0}^{1/2}+1 \right\} \\ 
& +\calE_{0}^{1/2}\left\| ^{(4)}F(t) \right\|_{{{L}^{\infty }}}^{{}}\left( \left\| \partial \phi  \right\|_{{{L}^{\infty }}}^{2}\left\| \phi  \right\|_{{{L}^{\infty }}}^{{}}+\left\| \partial \phi  \right\|_{{{L}^{\infty }}}^{{}} \right)+\left\| \phi  \right\|_{{{L}^{\infty }}}  ~ ,
\end{split}
\end{equation}
\begin{equation}
\begin{split}
{W}\left( t \right)&={{\left\| D\phi  \right\|}_{{{L}^{\infty }}}}Y\left( t \right)\left( {{\left\| \partial \phi  \right\|}_{{{L}^{\infty }}}}\calE_{0}^{1/2}+{{\left\| \phi  \right\|}_{{{L}^{\infty }}}}+{{\left\| \partial \Psi  \right\|}_{{{L}^{\infty }}}}\calE_{0}^{1/2}{{\left\| \partial \phi  \right\|}_{{{L}^{\infty }}}} \right) \\ 
& +{{\left\| ^{(4)}F(t) \right\|}_{{{L}^{\infty }}}}{{\left\| \phi  \right\|}_{{{L}^{\infty }}}}{{\left\| \partial \phi  \right\|}_{{{L}^{\infty }}}}\left( {{\left\| \partial \phi  \right\|}_{{{L}^{\infty }}}}\calE_{0}^{1/2}+{{\left\| \phi  \right\|}_{{{L}^{\infty }}}}\calE_{0}^{1/2}+{{\left\| \partial \Psi  \right\|}_{{{L}^{\infty }}}}\calE_{0}^{1/2}{{\left\| \partial \phi  \right\|}_{{{L}^{\infty }}}} \right) \\ 
& +{{\left\| D\phi (t) \right\|}_{{{L}^{\infty }}}}{{\left\| \partial \Psi  \right\|}_{{{L}^{\infty }}}}\tilde{Z}\calE_{0}^{1/2}+\left\| D\phi  \right\|_{{{L}^{\infty }}}^{{}}\left\| \phi  \right\|_{{{L}^{\infty }}}^{{}}\left\{ \hat{Z}\left\| \partial \phi  \right\|_{{{L}^{\infty }}}^{{}}\calE_{0}^{1/2}\left. +\tilde{Z} \right\} \right. \\ 
& +{{\left\| \partial \phi  \right\|}_{{{L}^{\infty }}}}Y(t)\left( {{\left\| \phi  \right\|}_{{{L}^{\infty }}}}+\calE_{0}^{1/2}\left\| \phi  \right\|_{{{L}^{\infty }}}^{2}+\calE_{0}^{1/2}{{\left\| \partial \phi  \right\|}_{{{L}^{\infty }}}}{{\left\| \phi  \right\|}_{{{L}^{\infty }}}}+{{\left\| D\phi (t) \right\|}_{{{L}^{\infty }}}}\calE_{0}^{1/2} \right) \\ 
& +{{\left\| ^{(4)}F(t) \right\|}_{{{L}^{\infty }}}}\left\| \partial \phi  \right\|_{{{L}^{\infty }}}^{2}\left\| \phi  \right\|_{{{L}^{\infty }}}^{3}\left( \left\| \partial \phi  \right\|_{{{L}^{\infty }}}^{2}+{{\left\| \phi  \right\|}_{{{L}^{\infty }}}} \right)\calE_{0}^{1/2}+\left\| \partial \phi  \right\|_{{{L}^{\infty }}}^{2}\left\| \phi  \right\|_{{{L}^{\infty }}}^{2}    ~ ,
\end{split}
\end{equation}
\begin{equation}
\begin{split}
{P}\left( t \right)&=\left\| \phi  \right\|_{{{L}^{\infty }}}^{2}\left\| \partial \phi  \right\|_{{{L}^{\infty }}}^{2}+\tilde{Z}\left\| D\phi (t) \right\|_{{{L}^{\infty }}}^{{}}\left\| \partial \phi  \right\|_{{{L}^{\infty }}}^{{}}\left\| \phi  \right\|_{{{L}^{\infty }}}^{{}}+\left\| ^{(4)}F(t) \right\|_{{{L}^{\infty }}}\left\| \partial \phi  \right\|_{{{L}^{\infty }}}\left\| \phi  \right\|_{{{L}^{\infty }}}^{2} \\ 
& +\left\| ^{(4)}F(t) \right\|_{{{L}^{\infty }}}^{{}}\left\| \partial \phi  \right\|_{{{L}^{\infty }}}+\left\| \phi  \right\|_{{{L}^{\infty }}}^{{}}T\left( t \right)+Y\left( t \right)\left( T\left( t \right)+{{\left\| \phi  \right\|}_{{{L}^{\infty }}}}+{{\left\| A \right\|}_{{{L}^{\infty }}}}+1 \right) \\ 
& +\left\| \phi  \right\|_{{{L}^{\infty }}}S\left( t \right)+\left\| \phi  \right\|_{{{L}^{\infty }}}\mathcal{Z}\left( \left\| \phi  \right\|_{{{L}^{\infty }}} \right)+Y\left( t \right)S\left( t \right)+\calE_{0}^{1/2}Y\left( t \right)\left( {{\left\| \partial \phi  \right\|}_{{{L}^{\infty }}}}+{{\left\| \phi  \right\|}_{{{L}^{\infty }}}} \right) \\ 
& +Y\left( t \right)\mathcal{Z}\left( \left\| \phi  \right\|_{{{L}^{\infty }}} \right)+\left\| \phi  \right\|_{{{L}^{\infty }}}^{2}\left\| \partial \phi  \right\|_{{{L}^{\infty }}}^{3}\left\| ^{(4)}F(t) \right\|_{{{L}^{\infty }}}\calE_{0}^{1/2}+\left\| D\phi  \right\|_{{{L}^{\infty }}}\left\| \partial \phi  \right\|_{{{L}^{\infty }}}\left\| \phi  \right\|_{{{L}^{\infty }}}^{{}}\calE_{0}^{1/2}Y\left( t \right) \\ 
& +\tilde{Z}\left\| \partial \phi  \right\|_{{{L}^{\infty }}}\left\| \phi  \right\|_{{{L}^{\infty }}}\left\| D\phi (t) \right\|_{{{L}^{\infty }}}\left( \calE_{0}^{1/2}+\calE_{0}^{1/2}\left( \left\| \phi  \right\|_{{{L}^{\infty }}}^{{}}{{\left\| \partial \phi  \right\|}_{{{L}^{\infty }}}}+\left\| \phi  \right\|_{{{L}^{\infty }}}^{2} \right) \right) \\ 
& +\left\| ^{(4)}F(t) \right\|_{{{L}^{\infty }}}^{{}}\left( \left\| \partial \Psi  \right\|_{{{L}^{\infty }}}^{{}}\left\| \partial \phi  \right\|_{{{L}^{\infty }}}^{2}\calE_{0}^{1/2}\left\| \phi  \right\|_{{{L}^{\infty }}}+\left\| \partial \phi  \right\|_{{{L}^{\infty }}}^{2}\left\| \phi  \right\|_{{{L}^{\infty }}}^{{}}\calE_{0}^{1/2} \right) \\ 
& +\left\| ^{(4)}F(t) \right\|_{{{L}^{\infty }}}^{{}}\left( \left\| \partial \Psi  \right\|_{{{L}^{\infty }}}^{{}}\left( \left\| \partial \phi  \right\|_{{{L}^{\infty }}}^{{}}\calE_{0}^{1/2}+\left\| \phi  \right\|_{{{L}^{\infty }}} \right)+\left\| \partial \Psi  \right\|_{{{L}^{\infty }}}^{2}\left\| \partial \phi  \right\|_{{{L}^{\infty }}}\calE_{0}^{1/2} \right)    ~ ,
\end{split}
\end{equation}
\begin{equation}
U=S\left( t \right)+T\left( t \right)+\mathcal{Z}\left( \left\| \phi  \right\|_{{{L}^{\infty }}}^{{}} \right),
\end{equation}
where
\begin{equation}
\tilde{Z}=\sum\limits_{n=1}^{N}{\frac{\left( n+2 \right)}{n+1}{{b}_{n}}\left\| \phi  \right\|_{{{L}^{\infty }}}^{n+2}}+{{C}_{1}}\left\| \phi  \right\|_{{{L}^{\infty }}},
\end{equation}
\begin{equation}
\hat{Z}=\sum\limits_{n=0}^{N}{\frac{n+2}{n+1}{{b}_{n}}{{\left\| \phi  \right\|}^{n+1}}}+\sum\limits_{n=0}^{N}{\left( n+3 \right){{b}_{n}}{{\left\| \phi  \right\|}^{n+2}}}+{{C}_{1}},
\end{equation}
\begin{equation}
\begin{split}
{S}(t)&={{\left\| \partial \phi  \right\|}_{{{L}^{\infty }}}}\tilde{Z}\left( {{\left\| \phi  \right\|}_{{{L}^{\infty }}}}\left\| ^{(4)}F(t) \right\|_{{{L}^{\infty }}}^{{}}\calE_{0}^{1/2}+\chi\left( \left\| \phi  \right\|_{{{L}^{\infty }}}^{{}} \right) \right) \\ 
& +\calE_{0}^{1/2}{{\left\| \partial \phi  \right\|}_{{{L}^{\infty }}}}{{{\tilde{Z}}}^{2}}\left( 1+{{\left\| \phi  \right\|}_{{{L}^{\infty }}}} \right)\left( \left\| D\phi (t) \right\|_{{{L}^{\infty }}}^{{}}+{{\left\| \partial \phi  \right\|}_{{{L}^{\infty }}}} \right) \\ 
& +\calE_{0}^{1/2}Z\left( \left\| D\phi (t) \right\|_{{{L}^{\infty }}}^{{}}+1 \right)\left( {{\left\| \partial \phi  \right\|}_{{{L}^{\infty }}}}+{{\left\| \phi  \right\|}_{{{L}^{\infty }}}} \right) \\ 
& +\calE_{0}^{1/2}\left\| ^{(4)}F(t) \right\|_{{{L}^{\infty }}}^{{}}\left\| \partial \phi  \right\|_{{{L}^{\infty }}}^{{}}+\left\| D\phi (t) \right\|_{{{L}^{\infty }}}^{{}}\left\| \phi  \right\|_{{{L}^{\infty }}}^{2}\left\| \partial \phi  \right\|_{{{L}^{\infty }}}^{{}}\calE_{0}^{1/2}\left( 1+{{\left\| \phi  \right\|}_{{{L}^{\infty }}}} \right) \\ 
& +\hat{Z}\left\| \partial \phi  \right\|_{{{L}^{\infty }}}^{{}}\left( 1+{{\left\| \phi  \right\|}_{{{L}^{\infty }}}} \right){{\calE}_{0}}    ~ ,
\end{split}
\end{equation}
\begin{equation}
T\left( t \right)=\calE_{0}^{1/2}\left( 1+{{\left\| \phi  \right\|}_{{{L}^{\infty }}}} \right)\tilde{Z}+\left\| ^{(4)}F(t) \right\|_{{{L}^{\infty }}}^{{}}\left\| \phi  \right\|_{{{L}^{\infty }}}^{{}}+Z\left( \left\| D\phi (t) \right\|_{{{L}^{\infty }}}^{{}}+1 \right),
\end{equation}
with
\begin{equation}
{\mathcal{Z}} \left( \left\| \phi  \right\| \right) \le \begin{cases}
{{\left\| \phi  \right\|}_{{{L}^{\infty }}}}\sum\limits_{n=1}^{N}{\left( n-1 \right)\left\| \Psi  \right\|_{{{L}^{\infty }}}^{n-2}}\calE_{0}^{1/2}+\sum\limits_{n=1}^{N}{n\left\| \Psi  \right\|_{{{L}^{\infty }}}^{n-1}} ~ , &\text{if $V$ is of the form \eqref{potential1}} ~ ,\\
\left( 1+\calE_{0}^{1/2}{{\left\| \phi  \right\|}_{{{L}^{\infty }}}} \right) ~ , & \text{if $V$  is either of the form \eqref{potential2}  }  \\
& \text{ or \eqref{potential3}} ~ ,
\end{cases}
\end{equation}
\begin{equation}
\chi \left( \left\| \phi  \right\| \right) \le \begin{cases}
\calE_{0}^{1/2}\sum\limits_{n=1}^{N}{n\left\| \Psi  \right\|_{{{L}^{\infty }}}^{n-1}} ~ , &\text{if $V$ is of the form \eqref{potential1}} ~ ,\\
\calE_{0}^{1/2} ~ , & \text{if $V$  is either of the form \eqref{potential2} or \eqref{potential3}} ~ ,
\end{cases}
\end{equation}
Integrating the inequality, we get
\begin{equation}
\label{E1Inequality}
{{\calE}_{1}}\left( t \right)\le {{\calE}_{1}}\left( 0 \right)\exp \left( \int\limits_{0}^{t}{\left\{{X}\left( t \right)+{W}\left( t \right)+{P}\left( t \right)+{U}\left( t \right)+ \right\}dt} \right)\:.
\end{equation}
The right hand side of (\ref{E1Inequality}) is a mixed expression of $\left\| ^{(4)}\mathcal{F}(t) \right\|_{{{L}^{\infty }}}^{{}},\left\| \phi (t) \right\|_{{{L}^{\infty }}}^{{}},\left\| \partial \phi (t) \right\|_{{{L}^{\infty }}}^{{}},$ $\left\| A(t) \right\|_{{{L}^{\infty }}}^{{}}$, $\left\| D\phi (t) \right\|_{{{L}^{\infty }}}^{{}}$, and ${{\calE}_{0}}\left( t \right)$ which is a bounded with respect to time. Therefore, by applying the Gronwall inequality, we find that ${{\calE}_{1}}\left( t \right)$ also cannot blow up in a finite time. 

Thus we have proven the global existence of MKG system with general coupling in temporal gauge condition
\begin{namedtheorem}[Main]
	Let ${{u}_{0}}=\left({{A}_{i}^{\Sigma}}(0),{{E}_{i}^{\Sigma}}(0),\phi^{a}(0) ,{{\partial }_{t}}\phi^{a}(0)\right) $ be the initial data on ${{\left( {{H}_{2}}\times {{H}_{1}} \right)}^{2}}$ such that the initial flat energy function on (\ref{modE}) is finite. If the internal scalar manifold satisfies Lemma \ref{scalarManifoldLemma} and Assumption \ref{boundpotAssumptions}, the gauge couplings satisfy Assumption \ref{GaugeCouplingAssumptions}, and the scalar potential is of the form given in Assumption \ref{ScalarpotentialAssumptions},  then there exist a unique global solution $u\left( t \right)\in {{\left( {{H}_{2}}\times {{H}_{1}} \right)}^{2}}$ of MKG equation with general gauge couplings in temporal gauge which solves the corresponding equations (\ref{eom1}) and (\ref{eom2}) for all $t\in \left( 0 ,\infty  \right)$. 
\end{namedtheorem}

\section{Acknowledgments}

The work in this paper is supported by P2MI FMIPA ITB 2021, Riset KK ITB 2021, and Riset PDUPT Kemendikbudristek-ITB 2021.

%\appendix

\end{document}